\documentclass[11pt]{article}
\usepackage[left=1in, right=1in, top=1in, bottom=1in]{geometry}
\usepackage{CJK}
\usepackage{latexsym,bm}
\usepackage{xypic}
\usepackage{amsmath}
\usepackage{wasysym}
\usepackage{indentfirst}
\usepackage{amssymb}
\usepackage{dsfont}
\usepackage{amsthm}
\begin{document}
\newcommand{\fr}[2]{\frac{\;#1\;}{\;#2\;}}
\newtheorem{theorem}{Theorem}[section]
\newtheorem{lemma}{Lemma}[section]
\newtheorem{proposition}{Proposition}[section]
\newtheorem{corollary}{Corollary}[section]
\newtheorem{conjecture}{Conjecture}[section]
\newtheorem{remark}{Remark}[section]
\newtheorem{definition}{Definition}[section]
\newtheorem{example}{Example}[section]
\newtheorem{notation}{Notation}[section]
\numberwithin{equation}{section}
\newcommand{\Aut}{\mathrm{Aut}\,}
\newcommand{\CSupp}{\mathrm{CSupp}\,}
\newcommand{\Supp}{\mathrm{Supp}\,}
\newcommand{\rank}{\mathrm{rank}\,}
\newcommand{\col}{\mathrm{col}\,}
\newcommand{\len}{\mathrm{len}\,}
\newcommand{\leftlen}{\mathrm{leftlen}\,}
\newcommand{\rightlen}{\mathrm{rightlen}\,}
\newcommand{\length}{\mathrm{length}\,}
\newcommand{\bin}{\mathrm{bin}\,}
\newcommand{\wt}{\mathrm{wt}\,}
\newcommand{\Wt}{\mathrm{Wt}\,}
\newcommand{\diff}{\mathrm{diff}\,}
\newcommand{\lcm}{\mathrm{lcm}\,}
\newcommand{\GL}{\mathrm{GL}\,}
\newcommand{\SJ}{\mathrm{SJ}\,}
\newcommand{\LG}{\mathrm{LG}\,}
\newcommand{\bij}{\mathrm{bij}\,}
\newcommand{\dom}{\mathrm{dom}\,}
\newcommand{\fun}{\mathrm{fun}\,}
\newcommand{\SUPP}{\mathrm{SUPP}\,}
\newcommand{\supp}{\mathrm{supp}\,}
\newcommand{\End}{\mathrm{End}\,}
\newcommand{\Hom}{\mathrm{Hom}\,}
\newcommand{\ran}{\mathrm{ran}\,}
\newcommand{\row}{\mathrm{row}\,}
\newcommand{\Mat}{\mathrm{Mat}\,}
\newcommand{\rk}{\mathrm{rk}\,}
\newcommand{\rs}{\mathrm{rs}\,}
\newcommand{\piv}{\mathrm{piv}\,}
\newcommand{\perm}{\mathrm{perm}\,}
\newcommand{\rsupp}{\mathrm{rsupp}\,}
\newcommand{\inv}{\mathrm{inv}\,}
\newcommand{\orb}{\mathrm{orb}\,}
\newcommand{\id}{\mathrm{id}\,}
\newcommand{\soc}{\mathrm{soc}\,}
\newcommand{\unit}{\mathrm{unit}\,}
\newcommand{\word}{\mathrm{word}\,}
\newcommand{\Sym}{\mathrm{Sym}\,}
\newcommand{\Jac}{\mathrm{Jac}\,}

\renewcommand{\thefootnote}{\fnsymbol{footnote}}

\title{Generalized Rank Weight and Extended Generalized Poset Weight Defined For Codes Over Rings: A Galois Connection Approach}
\author{Jianhua Zheng$^{1,2}$ \,\,\,\,\,\,Yang Xu$^1$ \,\,\,\,\,\,Yidan Chen$^2$ \,\,\,\,\,\, Haibin Kan$^{1,3}$\,\,\,\,\,\,Guangyue Han$^4$}

\maketitle

\renewcommand{\thefootnote}{\fnsymbol{footnote}}

\footnotetext{\hspace*{-6mm} \begin{tabular}{@{}r@{}p{16cm}@{}}
&This work has been supported by the National and Local Joint Laboratory of Cyberspace Security Technology (No. NLCS-202512-003).\\
$^1$ & Shanghai Key Laboratory of Intelligent Information Processing, College of Computer Science and Artificial Intelligence, Fudan University,
Shanghai 200433, China.\\
&Shanghai Engineering Research Center of Blockchain, Shanghai 200433, China. {E-mail:xuyyang@fudan.edu.cn, hbkan@fudan.edu.cn}\\
$^2$ & National and Local Joint Laboratory of Cyberspace Security Technology.\\
$^3$ &Yiwu Research Institute of Fudan University, Yiwu City, Zhejiang 322000, China. {E-mail:hbkan@fudan.edu.cn} \\
$^4$ & Department of Mathematics, Faculty of Science, The University of Hong Kong, Pokfulam Road, Hong Kong, China. {E-mail:ghan@hku.hk} \\
\end{tabular}}

\vskip 3mm

{\hspace*{-6mm}{\bf Abstract}---In this paper, we study generalized rank weights (GRWs) and extended generalized poset weight (EGPWs) of codes over rings via a Galois connection approach. First, we show that various coding-theoretic properties related to generalized weights, including security drops of a code employed in wire-tap channel of type II, connections between generalized weights of a Gabidulin code and its associated Delsarte code, (generalized) Singleton bound, MDS discrepancy of a code, characterizations of MDS, near MDS, $i$-MDS, MRD, near MRD, $i$-MRD, (dually) quasi-MRD codes as well as evasive property of subspaces, can be reformulated in terms of Galois connections. Next, we study GRWs and rank profiles defined for modules over principal ideal rings, especially those over chain rings. Generalizing GRWs defined for vector spaces over fields, we establish a singleton bound and a Wei-type duality theorem, characterize MRD, near MRD and dually quasi-MRD codes and determine their GRWs; moreover, we characterize $i$-MRD codes and establish a scattered bound for $(h,h)$-evasive codes over chain rings, generalizing counterpart result established for vector space over finite fields. Finally, we propose and study EGPWs and extended poset profiles defined for modules with a composition series, which in fact form a Galois connection. Generalizing EGPWs defined for modules over finite Galois rings, we establish a Wei-type duality theorem for modules over arbitrary quasi-Frobenius rings, which unifies the two Wei-type duality theorems derived in both \cite{32} and \cite{33}.

\section{Introduction}
In 1991, motivated by applications from information-theoretic security, Wei proposed and studied generalized Hamming
weights (GHWs) of linear codes. Roughly speaking, for a linear code $C$, the $r$-th GHW of $C$ is defined as the minimal Hamming weight of all the $r$-dimensional subcodes of $C$. It has been shown
in \cite{34} that GHWs characterize both the performance of a code on the wire-tap channel of type II and the performance of the code as a $t$-resilient function. GHWs have been widely used to gauge the security performances of linear codes for secret sharing, secure network coding or distributed data storage, see, e.g., \cite{5,15,24} and references therein.

Theoretically, GHWs have been used to characterize various families of codes such as $r$-th rank MDS codes ($r$-MDS codes for short), $A^{s}$MDS codes, and dually AMDS codes (or equivalently, near MDS codes) (see \cite{14,34}). Some important algebraic and combinatorial properties of GHWs, such as monotonicity and Wei's duality theorem, provide powerful tools for studying these codes (see \cite{14,34} for more details). Recently in \cite{32}, Tang proposed and studied an extension of generalized Hamming weights (EGHWs) for codes over $\mathbb{Z}_{p^{m}}$. In \cite{32}, among others, the author establishes two Wei-type duality theorems and several bounds for $\mathbb{Z}_{p^{m}}$-linear codes.

Extending the notion of GHWs, generalized weights with respect to rank metric and poset metric have become topics of great interest.

Rank metric codes were introduced by Delsarte in \cite{12} and by Gabidulin in \cite{16}, and generalized rank weights (GRWs) have been proposed and studied for Gabidulin codes by Kurihara, Matsumoto and Uyematsu in \cite{23}, by Ducoat in \cite{13}, by Jurrius and Pellikaan in \cite{21}, and have been extended to Delsarte codes by Ravagnani in \cite{28}, and by Mart\'{\i}nez-Pe\~{n}as and Matsumoto in \cite{26}. GRWs of Gabidulin codes and those of Delsarte codes are closely related, as it has been shown in \cite{28,26} that GRWs of a Gabidulin code are determined by those of its associated Delsarte code and vice versa (see Section 3.2 for more details). Rank metric codes and GRWs have important applications in secure network coding and crisscross error correction (see \cite{23,26,29,31}), and they have nice mathematical properties such as monotonicity, Wei-type duality and (generalized) Singleton bounds (see \cite{7,8,17,18} and references therein). Similar to GHWs, various families of rank metric codes such as $i$-MRD codes, $A^{s}$MRD codes, dually AMRD codes, near MRD codes and (dually) quasi-MRD codes have been studied by using GRWs (see \cite{7,8,25} and references therein). Recently, Marino, Neri and Trombetti have shown in \cite{25} that the evasive property of subspaces, a notion widely studied in combinatorics and finite geometry (see \cite{3,9} and references therein), can be reformulated via GRWs of Gabidulin codes.

Rank metric has also been extended to codes over rings. In \cite{22}, Kamche and Mouaha have proposed and studied rank metric codes over commutative principle ideal rings and explore their constructions and various applications. More recently in \cite{4}, Blanco-Chac\'{o}n, Boix, Greferath and Hieta-Aho have derived MacWilliams identities for rank metric codes over finite commutative chain rings, which has been further used to show that the dual of an MRD code is again MRD, generalizing well-known results for codes over finite fields.

Poset metric was first introduced by Brualdi, Graves and Lawrence in \cite{6}, and generalized poset weights (GPWs) have been proposed and studied by Moura and Firer in \cite{27}. In \cite{27}, the authors proved that similar to GHWS, GPWs have the properties of monotonicity and Wei-type duality, and they also generalized the notion of MDS discrepancy (see \cite{34}) from Hamming metric codes to poset codes. Various families of codes have been explored in the context of poset metric. For example, MDS poset codes and near-MDS poset codes have been proposed and studied by Hyun and Kim in \cite{20}, and by Barg and Purkayastha in \cite{2}, respectively. More recently in \cite{33}, Wang, Cao and Luo proposed extended generalized poset weights (EGPWs) for linear codes over Galois rings, which further generalizes EGHWs for $\mathbb{Z}_{p^{m}}$-linear codes. Among others, the authors established monotonicity, Singleton bound and two Wei-type duality theorems, and moreover, they explored the application of EGPWs in wire-tap channel of type II.

In this paper, we study GRWS and EGPWs defined for codes over rings. To make our results more general, the rings we consider can be both infinite and non-commutative. Our approach is based on Galois connections, a classical notion in combinatorics and algebra (see \cite{11}). This is based on the fact that GHWs and the dimension-length profiles (DLPs) of a code form a Galois connection, which in fact holds true for generalized weights and profiles with respect to various metrics (see \cite{5,15,23,26,35} and references therein). Hence Galois connections may provide a unified way to study some coding-theoretic properties that related to generalized weights. For example, a general Wei-type duality theorem for Galois connections has been established in \cite{35}.

The main contribution of the paper can be summarized as follows.

In Section 3, we reformulate various coding-theoretic properties in terms of Galois connections. More specifically, we establish some general properties of Galois connections, which can be used to characterize the security performance of a code in wire-tap channel of type II and network coding scheme (Section 3.1), the relations between GRWs of a Gabidulin code and those of its associated Delsarte code (Section 3.2), Singleton bound (Section 3.3), generalized Singleton bound (Sections 3.4 and 3.7), MDS/MRD/near MDS/near MRD/dually quasi-MRD codes (Sections 3.5, 3.9 and 3.10), MDS discrepancy of a code (Section 3.8) and $(h,r)$-evasiveness of subspaces (Section 3.11). Many results in this section are derived from a duality theorem (Theorem 3.5), which may be regarded as an analogue of the Wei-type duality theorem for Galois connections (Section 3.12).

In Section 4, we study GRWS and profiles for codes over principal left/right ideal rings. More specifically, in Section 4.1, we first characterize GRWs and profiles of a code and show that they form a Galois connection (Propositions 4.1, 4.2 and Theorem 4.1). Then, we define MRD, quasi-MRD (QMRD), near MRD (NMRD), $i$-MRD codes and $(h,r)$-evasiveness of a code, and derive some basic properties. In Sections 4.2-4.7, we focus on codes over chain rings. In Section 4.2, we derive a Wei-type duality theorem (Theorem 4.2); in Sections 4.3-4.6, we characterize MRD, dually QMRD, NMRD and $i$-MRD codes and determine the GRWs of the first three classes of codes (Theorems 4.3-4.6); in Section 4.7, we characterize $(h,r)$-evasiveness of a code in terms of GRWs of its dual code, and derive a scattered bound for $(h,h)$-evasive codes (Theorems 4.7 and 4.8), generalizing counterpart results established for evasive subspaces and Gabidulin codes in \cite{9,25}.

In Section 5, we propose and study EGPWs and profiles for codes over modules with a composition series. EGPWs and profiles of a code are determined by a poset together with a tuple of right/left ideals. This generalizes the notions of EGHWS for $\mathbb{Z}_{p^{m}}$-linear codes and EGPWs for Galois ring linear codes proposed in \cite{32,33}. In Sections 5.1 and 5.2, we characterize EGPWs and profiles of a code and show that they form a Galois connection (Propositions 5.1, 5.2 and Theorem 5.1). In Section 5.3, we establish a general Wei-type duality theorem for codes over quasi-Frobenius rings (Theorem 5.2). Our result unifies and generalizes the two Wei-type duality theorems established in \cite{33} as well as counterpart results in \cite{32}.

\section{Preliminaries}
We begin with a few notations and terminologies that will be used frequently in the rest of the paper. First of all, for any $a,b\in\mathbb{Z}$, let $[a,b]$ denote the set of all the integers between $a$ and $b$, i.e.,
$$[a,b]=\{i\in\mathbb{Z}\mid a\leqslant i\leqslant b\}.$$

Next, consider a non-empty finite set $\Omega$ and a poset $\mathbf{P}=(\Omega,\preccurlyeq_{\mathbf{P}})$. For any $I\subseteq\Omega$, $I$ is referred to as an \textit{ideal} of $\mathbf{P}$ if for any $v\in I$ and $u\in\Omega$, $u\preccurlyeq_{\mathbf{P}}v$ implies $u\in I$. We let $\mathcal{I}(\mathbf{P})$ denote the set of all the ideals of $\mathbf{P}$. For any $B\subseteq\Omega$, let $\langle B\rangle_{\mathbf{P}}=\{u\in\Omega\mid\exists~v\in B~s.t.~u\preccurlyeq_{\mathbf{P}}v\}$ denote the ideal of $\mathbf{P}$ generated by $B$. The dual poset of $\mathbf{P}$, denoted by $\mathbf{\overline{P}}=(\Omega,\preccurlyeq_{\mathbf{\overline{P}}})$, is defined as
$$u\preccurlyeq_{\mathbf{\overline{P}}}v\Longleftrightarrow v\preccurlyeq_{\mathbf{P}}u.$$
For an abelian group $M$ and any $D\subseteq M^{\Omega}$, the \textit{Hamming support of $D$} is defined as
\begin{equation}\chi(D)\triangleq\{i\in\Omega\mid\exists~\alpha\in D~s.t.~\alpha_i\neq0\},\end{equation}
and moreover, $\langle\chi(D)\rangle_{\mathbf{P}}$ is referred to as the \textit{$\mathbf{P}$-support of $D$}.

\subsection{Basics on Galois connection}
\setlength{\parindent}{2em}
In this subsection, we collect some basic properties of Galois connections between subsets of $\mathbb{Z}$ as well as some examples of Galois connections arising naturally in coding theory. We begin with the following definition, which is a special case of [11, Definition 7.23].

\setlength{\parindent}{0em}
\begin{definition}
Let $P,Q\subseteq\mathbb{Z}$, and let $\varphi:P\longrightarrow Q$ and $\psi:Q\longrightarrow P$. We say that $(\varphi,\psi)$ is a Galois connection between $P$ and $Q$ if for any $(a,b)\in P\times Q$, it holds that $\varphi(a)\leqslant b\Longleftrightarrow a\leqslant\psi(b)$.
\end{definition}

\setlength{\parindent}{2em}
Definition 2.1 implies that both $\varphi$ and $\psi$ preserve the order $\leqslant$, as detailed in the following remark.

\setlength{\parindent}{0em}
\begin{remark}
For $P,Q\subseteq\mathbb{Z}$, $\varphi:P\longrightarrow Q$ and $\psi:Q\longrightarrow P$, by [11, Lemma 7.26], $(\varphi,\psi)$ is a Galois connection between $P$ and $Q$ if and only if $\varphi(a)\leqslant\varphi(c)$ for all $a,c\in P$ with $a\leqslant c$, $\psi(b)\leqslant\psi(d)$ for all $b,d\in Q$ with $b\leqslant d$, $a\leqslant\psi(\varphi(a))$ for all $a\in P$, and $\varphi(\psi(b))\leqslant b$ for all $b\in Q$. Moreover, if $(\varphi,\psi)$ is a Galois connection between $P$ and $Q$, then by [11, 7.33], $\varphi$ and $\psi$ is determined by each other via the formula $\varphi(a)=\min\{b\in Q\mid a\leqslant\psi(b)\}$ and $\psi(d)=\max\{c\in P\mid \varphi(c)\leqslant d\}$.
\end{remark}

\setlength{\parindent}{2em}
The following lemma follows from [35, Lemma 2.2], and will be used frequently in the rest of the paper.

\begin{lemma}
Let $P$ and $Q$ be non-empty finite subsets of $\mathbb{Z}$, $X$ be a set, $f:X\longrightarrow P$ and $g:X\longrightarrow Q$ such that $\max(P)\in f[X]$, $\min(Q)\in g[X]$. Define $\varphi:P\longrightarrow Q$ and $\psi:Q\longrightarrow P$ as
$$\varphi(a)=\min\{g(u)\mid u\in X,a\leqslant f(u)\},$$
$$\psi(b)=\max\{f(u)\mid u\in X,g(u)\leqslant b\}.$$
Then, $(\varphi,\psi)$ is a Galois connection between $P$ and $Q$.
\end{lemma}

\setlength{\parindent}{2em}
In the following four examples, we use Lemma 2.2 to obtain Galois connections arising from generalized weights and profiles with respect to Hamming metric, poset metric and rank metric.

\begin{example}(GHWs and DLPs of linear codes)
Let $\mathbb{F}$ be a field, $m\in\mathbb{Z}^{+}$, and let $C\subseteq\mathbb{F}^{m}$ be a $k$-dimensional linear code. Following \cite{34}, for any $r\in[0,k]$, the $r$-th GHW of $C$ is defined as
\begin{eqnarray*}
\begin{split}
\mathbf{d}_r(C)&\triangleq\min\{|\chi(D)|\mid \text{$D$ is an $r$-dimensional subcode of $C$}\}\\
&=\min\{|J|\mid J\subseteq[0,m],\dim_{\mathbb{F}}(C\cap\delta(J))\geqslant r\},
\end{split}
\end{eqnarray*}
and following \cite{15}, for any $l\in[0,m]$, the $l$-th dimension/length profile (DLP) of $C$ is defined as
\begin{eqnarray*}
\begin{split}
\mathbf{K}_l(C)&\triangleq\max\{\dim_{\mathbb{F}}(C\cap\delta(J))\mid J\subseteq[0,m],|J|=l\}\\
&=\max\{\dim_{\mathbb{F}}(C\cap\delta(J))\mid J\subseteq[0,m],|J|\leqslant l\},
\end{split}
\end{eqnarray*}
where for any $J\subseteq[0,m]$,
$$\delta(J)=\{\alpha\in\mathbb{F}^{m}\mid\forall~i\in[0,m]-J:\alpha_i=0\}.$$
By Lemma 2.1, $\varphi:[0,k]\longrightarrow[0,m]$ defined as $\varphi(r)=\mathbf{d}_r(C)$ and $\psi:[0,m]\longrightarrow[0,k]$ defined as $\psi(l)=\mathbf{K}_l(C)$ form a Galois connection between $[0,k]$ and $[0,m]$.
\end{example}

\begin{example}(Generalized poset weights and profiles of linear codes)
Let $\mathbb{F}$ be a field, $\Omega$ be a non-empty set finite set with $|\Omega|=m$, $\mathbf{P}=(\Omega,\preccurlyeq_{\mathbf{P}})$ be a poset, and $C\subseteq\mathbb{F}^{\Omega}$ be a $k$-dimensional linear code. Following \cite{27}, for any $r\in[0,k]$, the $r$-th generalized $\mathbf{P}$-weight of $C$ is defined as
\begin{eqnarray*}
\begin{split}
\mathbf{d}_r^{\mathbf{P}}(C)&\triangleq\min\{|\langle\chi(D)\rangle_{\mathbf{P}}|\mid \text{$D$ is an $r$-dimensional subcode of $C$}\}\\
&=\min\{|J|\mid J\in\mathcal{I}(\mathbf{P}),\dim_{\mathbb{F}}(C\cap\delta(J))\geqslant r\},
\end{split}
\end{eqnarray*}
and for any $l\in[0,m]$, define the $l$-th $\mathbf{P}$-profile of $C$ as
\begin{eqnarray*}
\begin{split}
\mathbf{K}_l^{\mathbf{P}}(C)&\triangleq\max\{\dim_{\mathbb{F}}(C\cap\delta(J))\mid J\in\mathcal{I}(\mathbf{P}),|J|=l\}\\
&=\max\{\dim_{\mathbb{F}}(C\cap\delta(J))\mid J\in\mathcal{I}(\mathbf{P}),|J|\leqslant l\},
\end{split}
\end{eqnarray*}
where $\delta(J)$ is defined the same as in Example 2.1. By Lemma 2.1, $\varphi:[0,k]\longrightarrow[0,m]$ and $\psi:[0,m]\longrightarrow[0,k]$ defined as $\varphi(r)=\mathbf{d}_r^{\mathbf{P}}(C)$ and $\psi(l)=\mathbf{K}_l^{\mathbf{P}}(C)$ form a Galois connection between $[0,k]$ and $[0,m]$. We remark that when $\mathbf{P}$ is an anti-chain, then $\mathbf{d}_r^{\mathbf{P}}(C)$ and $\mathbf{K}_l^{\mathbf{P}}(C)$ boil down to $\mathbf{d}_r(C)$ and $\mathbf{K}_l(C)$ in Example 2.1, respectively.
\end{example}

\setlength{\parindent}{2em}
\begin{example}(GRWs and DIPs of Gabidulin codes)
Let $\mathbb{E}/\mathbb{F}$ be a field extension, $m\in\mathbb{Z}^{+}$, and let $\mathbb{E}^{[m]}$ denote the set of all the column vectors over $\mathbb{E}$ of length $m$. Any $\mathbb{E}$-subspace of $\mathbb{E}^{[m]}$ is referred to as a Gabidulin rank metric code. An $\mathbb{E}$-subspace $V\subseteq\mathbb{E}^{[m]}$ is said to be $\mathbb{E}/\mathbb{F}$-closed if $V$ has an $\mathbb{E}$-basis in $\mathbb{F}^{[m]}$. We let $\Gamma$ denote the set of all the $\mathbb{E}/\mathbb{F}$-closed $\mathbb{E}$-subspaces of $\mathbb{E}^{[m]}$. Noticing that for any $D\subseteq\mathbb{E}^{[m]}$, $\{V\in\Gamma\mid D\subseteq V\}$ contains a least member with respect to $\subseteq$, which we denote by $D^{\ast}$.

Let $C\subseteq\mathbb{E}^{[m]}$ be a Gabidulin code with $\dim_{\mathbb{E}}(C)=k$. Following \cite{21,23}, for any $r\in[0,k]$, the $r$-th GRW of $C$ is defined as
\begin{eqnarray*}
\begin{split}
\mathbf{d}_r^{G}(C)&\triangleq\min\{\dim_{\mathbb{E}}(D^{\ast})\mid \text{$D$ is an $r$-dimensional $\mathbb{E}$-subcode of $C$}\}\\
&=\min\{\dim_{\mathbb{E}}(V)\mid V\in\Gamma,\dim_{\mathbb{E}}(C\cap V)\geqslant r\},
\end{split}
\end{eqnarray*}
and following \cite{23}, for any $l\in[0,m]$, define the $l$-th dimension/intersection profile (DIP) of $C$ is defined as
\begin{eqnarray*}
\begin{split}
\mathbf{K}_l^{G}(C)&\triangleq\max\{\dim_{\mathbb{E}}(C\cap V)\mid V\in\Gamma,\dim_{\mathbb{E}}(V)=l\}\\
&=\max\{\dim_{\mathbb{E}}(C\cap V)\mid V\in\Gamma,\dim_{\mathbb{E}}(V)\leqslant l\}.
\end{split}
\end{eqnarray*}
By Lemma 2.1, $\varphi:[0,k]\longrightarrow[0,m]$ and $\psi:[0,m]\longrightarrow[0,k]$ defined as $\varphi(r)=\mathbf{d}_r^{G}(C)$ and $\psi(l)=\mathbf{K}_l^{G}(C)$ form a Galois connection between $[0,k]$ and $[0,m]$ (cf. [23, Lemma 9]).
\end{example}

\setlength{\parindent}{2em}
\begin{example}(GMWs and DRPs of Delsarte codes)
Let $\mathbb{F}$ be a field, $m,n\in\mathbb{Z}^{+}$, $\Mat_{m,n}(\mathbb{F})$ denote the set of all the matrices over $\mathbb{F}$ with $m$ rows and $n$ columns, and $\mathbb{F}^{[m]}$ denote the set of all the column vectors over $\mathbb{F}$ of length $m$. Any $\mathbb{F}$-subspace of $\Mat_{m,n}(\mathbb{F})$ is referred to as a Delsarte rank metric code. For any $\alpha\in\Mat_{m,n}(\mathbb{F})$, let $\col(\alpha)\subseteq\mathbb{F}^{[m]}$ denote the $\mathbb{F}$-subspace of $\mathbb{F}^{[m]}$ generated by all the columns of $\alpha$. Moreover, for any $D\subseteq\Mat_{m,n}(\mathbb{F})$, let $\CSupp(D)=\sum_{\alpha\in D}\col(\alpha)$.

Let $C\subseteq\Mat_{m,n}(\mathbb{F})$ be a Delsarte code with $\dim_{\mathbb{F}}(C)=k$. Following \cite{26}, for any $r\in[0,k]$, the $r$-th generalized matrix weight (GMW) of $C$ is defined as
\begin{eqnarray*}
\begin{split}
\mathbf{d}_r^{M}(C)&\triangleq\min\{\dim_{\mathbb{F}}(\CSupp(D))\mid \text{$D$ is an $r$-dimensional subcode of $C$}\}\\
&=\min\{\dim_{\mathbb{F}}(U)\mid \text{$U$ is an $\mathbb{F}$-subspace of $\mathbb{F}^{[m]}$, $\dim_{\mathbb{F}}(C\cap\delta(U))\geqslant r$}\},
\end{split}
\end{eqnarray*}
and for any $l\in[0,m]$, the $l$-th dimension/rank support profile (DRP) of $C$ is defined as
\begin{eqnarray*}
\begin{split}
\mathbf{K}_l^{M}(C)&\triangleq\max\{\dim_{\mathbb{F}}(C\cap\delta(U))\mid \text{$U$ is an $\mathbb{F}$-subspace of $\mathbb{F}^{[m]}$, $\dim_{\mathbb{F}}(U)=l$}\}\\
&=\max\{\dim_{\mathbb{F}}(C\cap\delta(U))\mid \text{$U$ is an $\mathbb{F}$-subspace of $\mathbb{F}^{[m]}$, $\dim_{\mathbb{F}}(U)\leqslant l$}\},
\end{split}
\end{eqnarray*}
where for any $\mathbb{F}$-subspace $U\subseteq\mathbb{F}^{[m]}$, $\delta(U)=\{\alpha\in\Mat_{m,n}(\mathbb{F})\mid\col(\alpha)\subseteq U\}$. By Lemma 2.1, $\varphi:[0,k]\longrightarrow[0,m]$ and $\psi:[0,m]\longrightarrow[0,k]$ defined as $\varphi(r)=\mathbf{d}_r^{M}(C)$ and $\psi(l)=\mathbf{K}_l^{M}(C)$ form a Galois connection between $[0,k]$ and $[0,m]$ (cf. [26, Proposition 14]). We also refer the reader to Ravagnani \cite{28} for definitions and results for Delsarte generalized weights, a different yet related notion proposed from an anti-code approach.
\end{example}

\setlength{\parindent}{2em}
The following lemma is a corollary of [35, Lemma 2.1 and Proposition 3.1], and will be used to derive monotonicity and generalized Singleton bounds in Sections 3, 4 and 5.

\setlength{\parindent}{0em}
\begin{lemma}
Let $k,m,w\in\mathbb{N}$, and let $(\varphi,\psi)$ be a Galois connection between $[0,k]$ and $[0,m]$. Then, we have $\varphi(0)=0$, $\psi(m)=k$. Moreover, the following two statements are equivalent to each other:

{\bf{(1)}}\,\,$\psi(l)-\psi(l-1)\leqslant w$ for all $l\in [1,m]$;

{\bf{(2)}}\,\,$\varphi(r)+1\leqslant\max\{\varphi(r+w),1\}$ for all $r\in[0,k-w]$.
\end{lemma}

\setlength{\parindent}{2em}
The following lemma follows from [35, Lemma 3.1] and some straightforward verification.

\setlength{\parindent}{0em}
\begin{lemma}
Let $(k,m)\in\mathbb{N}^{2}$, $w\in\mathbb{Z}^{+}$, and let $(\varphi,\psi)$ be a Galois connection between $[0,k]$ and $[0,m]$ such that $\psi(0)=0$ and $\psi(l)-\psi(l-1)\leqslant w$ for all $l\in[1,m]$. Then, we have $k\leqslant wm$, and there uniquely exists $\eta:[0,m]\longrightarrow[0,wm-k]$ defined as
$$\eta(l)=\psi(m-l)+wl-k.$$
Moreover, the following four statements hold:

{\bf{(1)}}\,\,$\eta(0)=0$, $\eta(m)=wm-k$, and for any $l\in[1,m]$, it holds that $0\leqslant\eta(l)-\eta(l-1)\leqslant w$;

{\bf{(2)}}\,\,There uniquely exists $\tau:[0,wm-k]\longrightarrow[0,m]$ such that $(\tau,\eta)$ is a Galois connection between $[0,wm-k]$ and $[0,m]$;

{\bf{(3)}}\,\,If $k\geqslant1$, then we have $\varphi(1)\geqslant1$ and $\{b\in[0,m]\mid \psi(b)=0\}=[0,\varphi(1)-1]$;

{\bf{(4)}}\,\,If $k\leqslant wm-1$, then we have $\tau(1)\geqslant1$ and $\{l\in[0,m]\mid\eta(l)=0\}=[0,\tau(1)-1]$.
\end{lemma}

\setlength{\parindent}{2em}
Lemma 2.3 is the starting point of most results established in Section 3. In coding-theoretic settings, $\varphi$ and $\psi$ are always equal to the generalized weights and profiles of a code, and $\tau$ and $\eta$ are equal to the generalized weights and profiles of the dual code, respectively, as illustrated in the following example.

\setlength{\parindent}{0em}
\begin{example}
{\bf{(1)}}\,\,Following Example 2.2, let $\varphi$ and $\psi$ denote the generalized $\mathbf{P}$-weights and $\mathbf{P}$-profiles of the $k$-dimension code $C$, respectively, and set $w=1$. Then, $\tau$ and $\eta$ coincide with the generalized $\overline{\mathbf{P}}$-weights and $\overline{\mathbf{P}}$-profiles of $C^{\bot}$ (the dual code of $C$ with respect to the standard inner product on $\mathbb{F}^{\Omega}$), respectively. More precisely, we have $\tau(s)=\mathbf{d}_s^{\overline{\mathbf{P}}}(C^{\bot})$ for all $s\in[0,m-k]$, and $\eta(l)=\mathbf{K}_l^{\overline{\mathbf{P}}}(C^{\bot})$ for all $l\in[0,m]$. In particular, if $\mathbf{P}$ is an anti-chain, then $\tau$ and $\eta$ become the GHWs and DLPs of $C^{\bot}$, respectively (see [15, Theorems 2, 3, 4]).

{\bf{(2)}}\,\,Following Example 2.3, let $\varphi$ and $\psi$ denote the GRWs and DIPs of the $k$-dimension Gabidulin code $C$, respectively, and set $w=1$. Then, $\tau$ and $\eta$ coincide with the GRWs and DIPs of $C^{\bot}$ (the dual code of $C$ with respect to the standard inner product on $\mathbb{E}^{[m]}$), respectively. More precisely, we have $\tau(s)=\mathbf{d}_s^{G}(C^{\bot})$ for all $s\in[0,m-k]$, and $\eta(l)=\mathbf{K}_l^{G}(C^{\bot})$ for all $l\in[0,m]$.

{\bf{(3)}}\,\,Following Example 2.4, let $\varphi$ and $\psi$ denote the GMWs and DRPs of the $k$-dimension Delsarte code $C$, respectively, and set $w=n$. Then, $\tau$ and $\eta$ coincide with the GMWs and DRPs of $C^{\bot}$ (the dual code of $C$ with respect to the trace inner product on $\Mat_{m,n}(\mathbb{F})$, see \cite{26}), respectively. More precisely, we have $\tau(s)=\mathbf{d}_s^{M}(C^{\bot})$ for all $s\in[0,mn-k]$, and $\eta(l)=\mathbf{K}_l^{M}(C^{\bot})$ for all $l\in[0,m]$.
\end{example}

\subsection{Basics on rings and modules}

\setlength{\parindent}{2em}
We first follow \cite{1,10} and collect some notations and terminologies for rings and modules. Let $R$ be a ring (which is always assumed to be associative with an multiplicative identity), and let $M$ be a left $R$-module. For any $A\subseteq M$, we write $A\leqslant_{R}M$ if $A$ is a left $R$-submodule of $M$. If $M$ is finitely generated, then we let $\rank_R(M)$ denote the minimal cardinality of a generator set of $M$. $M$ is said to be \textit{simple} if $M\neq\{0\}$ and $M$ has no left $R$-submodules other than $\{0\}$ and $M$. A \textit{composition series} of $M$ is a chain of left $R$-submodules $\{0\}=L_0\subseteq L_1\subseteq\cdots\subseteq L_{t-1}\subseteq L_{t}=M$ where $t\in\mathbb{N}$ and $L_i/L_{i-1}$ is a simple left $R$-module for all $i\in[1,t]$. If $M$ has a composition series $\{0\}=L_0\subseteq L_1\subseteq\cdots\subseteq L_{t-1}\subseteq L_{t}=M$, then we let $\len_{R}(M)\triangleq t$. By the Jordan-H\"{o}lder theorem, $t$ is independent of the choice of the composition series, and hence $\len_{R}(M)$ is well-defined. We note that all the above notions parallelly apply to right $R$-modules.

For any $A\subseteq R$, the \textit{right annihilator of $A$ in $R$} is defined as $A^{\dag}=\{b\in R\mid \forall~a\in A:ab=0\}$, and the \textit{left annihilator of $A$ in $R$} is defined as ${^{\dag}A}=\{b\in R\mid \forall~a\in A:ba=0\}$. We say that $R$ is \textit{quasi-Frobenius} if $R$ has a composition series as a left $R$-module, and for any left ideal $A$ of $R$ and right ideal $B$ of $R$, it holds that $^{\dag}(A^{\dag})=A$ and $({^{\dag}B})^{\dag}=B$.

Now following \cite{22,19}, $R$ is said to be a \textit{principal left ideal ring} if every left ideal of $R$ is a principal left ideal; moreover, $R$ is said to be a \textit{left chain ring} if the set of all the principal left ideals of $R$ is a finite chain under $\subseteq$ (i.e., $R$ has only finitely many principal left ideals, and for any two principal left ideals $I_1,I_2$ of $R$, it holds that either $I_1\subseteq I_2$ or $I_2\subseteq I_1$). Noticing that principal right ideal ring and right chain ring can be defined in a parallel fashion, we say that $R$ is a \textit{chain ring} if $R$ is both a left chain ring and a right chain ring. We note that a left chain ring is a principal left ideal ring, and a chain ring is necessarily quasi-Frobenius.

The following lemma follows from [19, Section 2] and the first assertion of [22, Proposition 3.2], and will be used frequently in Section 4. We note that the results in \cite{19,22} were established for finite rings (and the ring is assumed to be commutative in \cite{22}), and the counterpart results can be extended to infinite or non-commutative rings.

\setlength{\parindent}{0em}
\begin{lemma}
Suppose that $R$ is a principal left ideal ring and $M$ is a finitely generated left $R$-module, and let $B\leqslant_{R}M$. Then, $B$ is finitely generated with $\rank_R(B)\leqslant\rank_R(M)$. Further assume that $R$ is a left chain ring and $M$ is a free left $R$-module. Then, there exists a free left $R$-submodule $H\leqslant_{R}M$ such that $B\subseteq H$ and $\rank_R(H)=\rank_R(B)$. Moreover, for any $V\leqslant_{R}M$, it holds that $\rank_R(B\cap V)=\rank_R(H\cap V)$.
\end{lemma}

\setlength{\parindent}{2em}
Now consider a left $R$-module $M$, a right $R$-module $N$ and an \textit{$(R,R)$-bilinear map} $f:M\times N\longrightarrow R$, i.e., for any $u,v\in M$, $x,y\in N$ and $a,b\in R$, it holds that $f(u+v,x)=f(u,x)+f(v,x)$, $f(u,x+y)=f(u,x)+f(u,y)$, $f(au,xb)=af(u,x)b$. For any $A\subseteq M$, the \textit{right dual of $A$ with respect to $f$} is defined as
$$A^{\bot}\triangleq\{y\in N\mid \forall~x\in A:f(x,y)=0\},$$
and for any $B\subseteq N$, the \textit{left dual of $B$ with respect to $f$} is defined as
$${^{\bot}B}\triangleq\{x\in M\mid \forall~y\in B:f(x,y)=0\}.$$
Moreover, $f$ is said to be \textit{non-degenerate} if $M^{\bot}=\{0\}$ and ${^{\bot}N}=\{0\}$.

The following lemma will be used in Sections 4 and 5 for deriving Wei-type duality theorems.

\setlength{\parindent}{0em}
\begin{lemma}
{\bf{(1)}}\,\,Suppose that $R$ is quasi-Frobenius, $f$ is non-degenerate, and both $M$ and $N$ have a composition series. Then, for any $A\leqslant_{R}M$ and $B\leqslant_{R}N$, we have $^{\bot}(A^{\bot})=A$, $({^{\bot}B})^{\bot}=B$, $\len_{R}(A^{\bot})=\len_{R}(M)-\len_{R}(A)$ and $\len_R(A\cap({^{\bot}B}))-\len_R(A^{\bot}\cap B)=\len_{R}(A)-\len_{R}(B)$.

{\bf{(2)}}\,\,Suppose that $R$ is a chain ring, $f$ is non-degenerate, and both $M$ and $N$ are finitely generated free $R$-modules with $\rank_R(M)=\rank_R(N)=m$. Let $A\leqslant_{R}M$ and $B\leqslant_{R}N$ be free $R$-submodules. Then, $A^{\bot}$ is a free right $R$-submodule of $N$ with $\rank_R(A^{\bot})=m-\rank_R(A)$, ${^{\bot}B}$ is a free left $R$-submodule of $M$ with $\rank_R({^{\bot}B})=m-\rank_R(B)$, and it holds that
$$\rank_R(A\cap({^{\bot}B}))-\rank_R(A^{\bot}\cap B)=\rank_R(A)-\rank_R(B).$$
\end{lemma}

\begin{proof}
(1) is a corollary of [1, Theorem 30.1] or [10, Theorem 58.8], and (2) follows from [19, Section 2]. We omit the details of the verification.
\end{proof}

\section{Galois connection and related coding-theoretic properties}
\setlength{\parindent}{2em}
In this section, we establish some results for Galois connections and apply these results to coding-theoretic scenarios.

\subsection{Security drop}
\setlength{\parindent}{2em}
It was shown in [34, Corollary A] that GHWs characterize the performance of code on the wire-tap channel of type II. An analogue of this result for rank metric code was established in [28, Theorem 42] by showing that GRWs measure the worst-case security drops of a Gabidulin code when it is employed in the scheme proposed by Silva, Kschischang and Koetterin \cite{31}, which can be used to secure a network communication against an eavesdropper. Both of these results may be illustrated by the following proposition.

\setlength{\parindent}{0em}
\begin{proposition}
Let $(k,m)\in\mathbb{N}^{2}$, and let $(\varphi,\psi)$ be a Galois connection between $[0,k]$ and $[0,m]$. Then, the following two statements hold:

{\bf{(1)}}\,\,For any $r\in [0,k-1]$, it holds that $\{b\in[0,m]\mid \psi(b)=r\}=[\varphi(r),\varphi(r+1)-1]$; in particular, $r\in\ran(\psi)$ if and only if $\varphi(r)+1\leqslant\varphi(r+1)$;

{\bf{(2)}}\,\,For any $l\in [1,m]$, it holds that $\{a\in[0,k]\mid \varphi(a)=l\}=[\psi(l-1)+1,\psi(l)]$; in particular, $l\in\ran(\varphi)$ if and only if $\psi(l-1)+1\leqslant\psi(l)$.
\end{proposition}

\begin{proof}
{\bf{(1)}}\,\,For any $b\in[0,m]$, from $r\leqslant\psi(b)\Longleftrightarrow\varphi(r)\leqslant b$ and $\psi(b)\leqslant r\Longleftrightarrow b\leqslant \varphi(r+1)-1$, we deduce that $\psi(b)=r$ if and only if $\varphi(r)\leqslant b\leqslant \varphi(r+1)-1$, as desired.

{\bf{(2)}}\,\,For any $a\in[0,k]$, from $\varphi(a)\leqslant l\Longleftrightarrow a\leqslant \psi(l)$ and $l\leqslant \varphi(a)\Longleftrightarrow\psi(l-1)+1\leqslant a$, we deduce that $\varphi(a)=l$ if and only if $\psi(l-1)+1\leqslant a\leqslant \psi(l)$, as desired.
\end{proof}

\setlength{\parindent}{0em}
\begin{example}
{\bf{(1)}}\,\,Following Example 2.1 and consider a $k$-dimensional linear code $C\subseteq\mathbb{F}^{m}$. By (2) of Proposition 3.2, for any $l\in [1,m]$, $\mathbf{K}_{l-1}(C)+1\leqslant\mathbf{K}_l(C)$ if and only if $l=\mathbf{d}_r(C)$ for some $r\in[1,k]$. Alternatively speaking, the DLPs of $C$ increase at exactly the GHWs of $C$. When $C^{\bot}$ is employed in the wire-tap channel of type II, $\mathbf{K}_l(C)$ measures the information obtained
by the adversary with $l$ taps. Hence the drops of the adversary's equivocation occurs at exactly the GHWs of $C$, as has been shown in [34, Corollary A].

{\bf{(2)}}\,\,Following Example 2.3 and consider a Gabidulin code $C\subseteq\mathbb{E}^{[m]}$ with $\dim_{\mathbb{E}}(C)=k$. By (2) of Proposition 3.2, for any $l\in [1,m]$, $\mathbf{K}_{l-1}^{G}(C)+1\leqslant\mathbf{K}_l^{G}(C)$ if and only if $l=\mathbf{d}_r^{G}(C)$ for some $r\in[1,k]$. Alternatively speaking, the DIPs of $C$ increases at exactly the GRWs of $C$. When $C$ is employed in the rank-metric coding scheme  proposed in \cite{31}, the information that an eavesdropper can obtain listening at $l$ links of the channel is bounded by $\mathbf{K}_l^{G}(C)$. Moreover, following [28, Definition 41], the worst case security drops of $C$ are those $l's$ between $1$ and $m$ satisfying  $\mathbf{K}_{l-1}^{G}(C)+1\leqslant\mathbf{K}_l^{G}(C)$, which are exactly the GRWs of $C$, as has been shown in [28, Theorem 42].
\end{example}

\subsection{Relations between generalized weights of Gabidulin and Delsarte codes}

\setlength{\parindent}{2em}
It was shown in [28, Theorem 28] that GRWs of a Gabidulin code are determined by the Delsarte generalized weights of its associated Delsarte code and vice versa. An analogue of this result for GMWs (and the more general notion of relative generalized matrix weights or RGMWs, see [26, Definition 10]) was established in [26, Theorem 7]. Both of these results may be illustrated by the following proposition.

\setlength{\parindent}{0em}
\begin{proposition}
Let $k\in\mathbb{N}$, $Q\subseteq\mathbb{Z}$, and let $(\varphi,\psi)$ be a Galois connection between $[0,k]$ and $Q$. Fix $\theta\in\mathbb{Z}^{+}$. Define $f:[0,\theta k]\longrightarrow Q$ as $f(a)=\varphi\left(\left\lceil\frac{a}{\theta}\right\rceil\right)$, and define $g:Q\longrightarrow [0,\theta k]$ as $g(b)=\theta\psi(b)$. Then, $(f,g)$ is a Galois connection between $[0,\theta k]$ and $Q$. Moreover, if $h:[0,\theta k]\longrightarrow Q$ and $(h,g)$ is a Galois connection between $[0,\theta k]$ and $Q$, then we have $h(a)=\varphi\left(\left\lceil\frac{a}{\theta}\right\rceil\right)$ for all $a\in[0,\theta k]$.
\end{proposition}

\begin{proof}
Let $a\in[0,\theta k]$ and $b\in Q$. Noticing that $\psi(b)$ is an integer, we have
$$f(a)\leqslant b\Longleftrightarrow\varphi\left(\left\lceil\frac{a}{\theta}\right\rceil\right)\leqslant b\Longleftrightarrow\left\lceil\frac{a}{\theta}\right\rceil\leqslant\psi(b)\Longleftrightarrow\frac{a}{\theta}\leqslant\psi(b)\Longleftrightarrow a\leqslant\theta\psi(b)\Longleftrightarrow a\leqslant g(b),$$
which implies that $(f,g)$ is a Galois connection between $[0,\theta k]$ and $Q$, as desired. Finally, the ''moreover'' assertion follows from the ``moreover'' part of Remark 2.1.
\end{proof}

\setlength{\parindent}{2em}
\begin{example}
Let $\mathbb{E}/\mathbb{F}$ be a finite dimensional field extension with $\dim_{\mathbb{F}}(\mathbb{E})=n$, and fix an ordered $\mathbb{F}$-basis $(\tau_1,\dots,\tau_n)$ of $\mathbb{E}$. Let $m\in\mathbb{Z}^{+}$. For any $\alpha\in\mathbb{E}^{[m]}$, define $\mathcal{M}(\alpha)\in\Mat_{m,n}(\mathbb{F})$ as $\alpha_i=\sum_{j=1}^{n}\mathcal{M}(\alpha)_{(i,j)}\tau_j$ for all $i\in[1,m]$. Now for an $k$-dimensional Gabidulin code $C\subseteq\mathbb{E}^{[m]}$, the Delsarte code associated to $C$ is defined as $Q\triangleq\{\mathcal{M}(\alpha)\mid\alpha\in C\}$. By Examples 2.3 and 2.4, $(\mathbf{d}_r^{G}(C)\mid r\in[0,k])$ and $(\mathbf{K}_l^{G}(C)\mid l\in[0,m])$ form a Galois connection between $[0,k]$ and $[0,m]$, and $(\mathbf{d}_a^{M}(Q)\mid a\in[0,nk])$ and $(\mathbf{K}_l^{M}(Q)\mid l\in[0,m])$ form a Galois connection between $[0,nk]$ and $[0,m]$. Moreover, it has been shown in [26, Theorem 7] that $\mathbf{K}_l^{M}(Q)=n\mathbf{K}_l^{G}(C)$ for all $l\in[0,m]$ and $\mathbf{d}_a^{M}(Q)=\mathbf{d}_{_{\left\lceil\frac{a}{n}\right\rceil}}^{G}(C)$ for all $a\in[0,nk]$, where the latter can be derived from the former one and Proposition 3.2.
\end{example}

\subsection{Singleton bound}
\setlength{\parindent}{2em}
Throughout the rest of this section, let $(k,m)\in\mathbb{N}^{2}$, $w\in\mathbb{Z}^{+}$, and let $(\varphi,\psi)$ be a Galois connection between $[0,k]$ and $[0,m]$ such that $\psi(0)=0$ and
\begin{equation}\text{$\psi(l)-\psi(l-1)\leqslant w$ for all $l\in[1,m]$}. \end{equation}
Moreover, define $\eta:[0,m]\longrightarrow[0,wm-k]$ as
\begin{equation}\eta(l)=\psi(m-l)+wl-k,\end{equation}
and let $\tau:[0,wm-k]\longrightarrow[0,m]$ such that $(\tau,\eta)$ is a Galois connection between $[0,wm-k]$ and $[0,m]$.

In this subsection, we establish a Singleton bound for $k$, as detailed in the following proposition.

\setlength{\parindent}{0em}
\begin{proposition}
Suppose that $k\geqslant1$ and $d=\varphi(1)$. Then, for any $l\in[d-1,m]$, it holds that
\begin{equation}k-w(m-l)\leqslant\psi(l)\leqslant w(l-d+1).\end{equation}
In particular, we have
\begin{equation}k\leqslant w(m-d+1).\end{equation}
\end{proposition}

\begin{proof}
For any $l\in[d-1,m]$, from $\eta(m-l)=\psi(l)+w(m-l)-k\geqslant0$, we deduce that $k-w(m-l)\leqslant\psi(l)$; moreover, from $\psi(d-1)=0$ and (3.1), we deduce that $\psi(l)\leqslant w(l-(d-1))=w(l-d+1)$, as desired. Finally, (3.3) implies that $k=\psi(m)\leqslant w(m-d+1)$, as desired.
\end{proof}

\setlength{\parindent}{2em}
Now we show that if equality holds in (3.4), then $(\varphi,\psi)$ is uniquely determined.

\begin{theorem}
Suppose that $k\geqslant1$, $d=\varphi(1)$ and $k=w(m-d+1)$. Then, we have $\psi(l)=w(l-d+1)$ for all $l\in[d-1,m]$, and $\varphi(a)=\left\lceil\frac{a}{w}\right\rceil+d-1$ for all $a\in[1,k]$.
\end{theorem}

\begin{proof}
For any $l\in[d-1,m]$, from $k-w(m-l)=w(l-d+1)$ and (3.3), we deduce that $\psi(l)=w(l-d+1)$, as desired. Next, consider $a\in[1,k]$. Let $b\triangleq\varphi(a)\geqslant d$. By (2) of Proposition 3.1, we have $\psi(b-1)+1\leqslant a\leqslant\psi(b)$, which implies that $w(b-d)+1\leqslant a\leqslant w(b-d+1)$. It then follows that $\lceil\frac{a}{w}\rceil=b-d+1$ and hence $\varphi(a)=b=\lceil\frac{a}{w}\rceil+d-1$, as desired.
\end{proof}

\setlength{\parindent}{2em}
In the following example, we use (3.4) to recover Singleton bounds for codes endowed with various metrics.

\setlength{\parindent}{0em}
\begin{example}
{\bf{(1)}}\,\,Following (1) of Example 2.5 and consider a linear code $C\subseteq\mathbb{F}^{\Omega}$ with $\dim_{\mathbb{F}}(C)=k\geqslant1$. Then, $d$ becomes the minimal $\mathbf{P}$-weight of $C$ (see \cite{20}), and (3.4) recovers the Singleton bound for $\mathbf{P}$-codes, which has been established in [20, Corollary 2.2]. When $k=m-d+1$, then $C$ is an MDS $\mathbf{P}$-code (see [20, Definition 2.3]). Theorem 3.1 shows that the generalized $\mathbf{P}$-weights and profiles of an MDS $\mathbf{P}$-code is completely determined by the code length $m$ and dimension $k$. In particular, if $\mathbf{P}$ is an anti-chain, then $d$ is equal to the minimal Hamming weight of $C$, and (3.4) recovers the classical Singleton bound.

{\bf{(2)}}\,\,Following (2) of Example 2.5 and consider a Gabidulin code $C\subseteq\mathbb{E}^{[m]}$ with $\dim_{\mathbb{E}}(C)=k\geqslant1$. Then, $d$ is equal to the minimal rank weight of $C$ (see [17, Definition 3.3]), and (3.4) becomes the Singleton bound for Gabidulin codes (see [17, Remark 3.6]). If $k=m-d+1$, then $C$ is an MRD code (see [17, Definition 3.8]), whose GRWs and DIPs are completely determined by $m$ and $k$, as has been shown in Theorem 3.1.

{\bf{(3)}}\,\,Following (3) of Example 2.5 and consider a Delsarte code $C\subseteq\Mat_{m,n}(\mathbb{F})$ with $\dim_{\mathbb{F}}(C)=k\geqslant1$. Then, $d$ is equal to the minimal rank weight of $C$ (see [17, Definition 3.1]), and (3.4) becomes the Singleton bound for Delsarte codes (cf. [17, Theorem 3.5]). If $k=n(m-d+1)$, then $C$ is an MRD Delsarte code (see [17, Definition 3.8]), whose GMWs and DRPs are completely determined by $(m,n,k)$, as has been shown in Theorem 3.1.
\end{example}

\subsection{Generalized Singleton bound}
\setlength{\parindent}{2em}
Now we establish a generalized Singleton bound for $\varphi$. The following theorem is the main result of this subsection.

\setlength{\parindent}{0em}
\begin{theorem}
{\bf{(1)}}\,\,For any $r\in[0,k-w]$, it holds that $\varphi(r)+1\leqslant\varphi(r+w)$.

{\bf{(2)}}\,\,For any $a,c\in[0,k]$ with $a\leqslant c$, it holds that
\begin{equation}\varphi(a)+\left\lfloor\frac{c-a}{w}\right\rfloor\leqslant\varphi(c).\end{equation}

{\bf{(3)}}\,\,For any $a\in[0,k]$, it holds that
\begin{equation}\varphi(a)\leqslant m-\left\lfloor\frac{k-a}{w}\right\rfloor.\end{equation}

{\bf{(4)}}\,\,Let $a\in[0,k]$ such that $\varphi(a)=m-\left\lfloor\frac{k-a}{w}\right\rfloor$, and let $c\in[a,k]$ with $c\equiv a~(\bmod~w)$. Then, it holds that $\varphi(c)=m-\left\lfloor\frac{k-c}{w}\right\rfloor$.
\end{theorem}

\begin{proof}
{\bf{(1)}}\,\,For any $r\in[0,k-w]$, by $r+w\geqslant1$, we have $\varphi(r+w)\geqslant1$, which, together with (3.1) and Lemma 2.2, implies that $\varphi(r)+1\leqslant\max\{\varphi(r+w),1\}=\varphi(r+w)$, as desired.

{\bf{(2)}}\,\,By $a\leqslant a+\left\lfloor\frac{c-a}{w}\right\rfloor w\leqslant c$ and (1), we have $\varphi(a)+\left\lfloor\frac{c-a}{w}\right\rfloor\leqslant\varphi(a+\left\lfloor\frac{c-a}{w}\right\rfloor w)\leqslant\varphi(c)$, as desired.

{\bf{(3)}}\,\,By (2), we have $\varphi(a)+\lfloor\frac{k-a}{w}\rfloor\leqslant\varphi(k)\leqslant m$, which establishes (3.6).

{\bf{(4)}}\,\,From (2) and $\frac{c-a}{w}\in\mathbb{N}$, we deduce that $\varphi(c)\geqslant\varphi(a)+\frac{c-a}{w}=m-\left\lfloor\frac{k-c}{w}\right\rfloor$; moreover, by (3), we have $\varphi(c)\leqslant m-\left\lfloor\frac{k-c}{w}\right\rfloor$, which further establishes the desired result.
\end{proof}

\setlength{\parindent}{2em}
We note that (1) of Theorem 3.2 establishes monotonicity of $\varphi$, and inequality (3.6), which follows from the more general inequality (3.5), may be regarded as the generalized Singleton bound for $\varphi$. When $k\geqslant1$, (3.6) recovers the Singleton bound (3.4) by setting $a=1$.

In the following example, we use Theorem 3.2 to recover several generalized Singleton bounds in coding theory.

\setlength{\parindent}{0em}
\begin{example}
Following (1) of Example 2.5, and consider a linear code $C\subseteq\mathbb{F}^{\Omega}$ with $\dim_{\mathbb{F}}(C)=k$. Then, (1) of Theorem 3.2 implies that $\mathbf{d}_r^{\mathbf{P}}(C)+1\leqslant\mathbf{d}_{r+1}^{\mathbf{P}}(C)$ for all $r\in[0,k-1]$, and (3.6) implies that $\mathbf{d}_r^{\mathbf{P}}(C)\leqslant m-k+r$ for all $r\in[0,k]$. These recover the monotonicity and generalized Singleton bound of generalized $\mathbf{P}$-weights, respectively (see [27, Propositions 1 and 2]). In particular, if $\mathbf{P}$ is an anti-chain, then we obtain the well-known monotonicity and generalized Singleton bound for GHWs (see [34, Theorem 1 and Corollary 1]). Similarly, following (2) and (3) of Example 2.5, (1) and (3) of Theorem 3.2 recover the monotonicity and generalized Singleton bound for GRWs of Gabidulin codes ([13, Theorem I.2 and Corollary II.4]) as well as those for GMWs of Delsarte codes (see [26, Proposition 43 and Theorem 5]), respectively.
\end{example}

\setlength{\parindent}{2em}
The following theorem gives monotonicity and generalized Singleton bound for $\tau$. The proof is similar to that of Theorem 3.2 and thus omitted.

\setlength{\parindent}{0em}
\begin{theorem}
We have $\tau(c)+1\leqslant\tau(c+w)$ for all $c\in[0,w(m-1)-k]$, and $\tau(a)\leqslant \left\lceil\frac{k+a}{w}\right\rceil$ for all $a\in[0,wm-k]$.
\end{theorem}

\setlength{\parindent}{2em}
We end this subsection with the following proposition, which improves (3.4) and gives sufficient conditions for $\varphi$ and $\tau$ to obtain the generalized Singleton bounds.

\setlength{\parindent}{0em}
\begin{proposition}
{\bf{(1)}}\,\,Suppose that $k\leqslant wm-1$ and $e=\tau(1)$. Then, for any $t\in[m+1-e,m]$, it holds that
\begin{equation}\psi(t)=k-w(m-t).\end{equation}
In particular, it holds that
\begin{equation}k=w(e-1)+\psi(m-e+1).\end{equation}
Moreover, for any $a\in[0,k]$ with $\varphi(a)\geqslant m+1-e$, it holds that
\begin{equation}\varphi(a)=m-\left\lfloor\frac{k-a}{w}\right\rfloor.\end{equation}

{\bf{(2)}}\,\,Suppose that $k\geqslant1$ and $d=\varphi(1)$. Then, for any $l\in[m-d+1,m]$, it holds that
\begin{equation}\eta(l)=wl-k.\end{equation}
In particular, it holds that
\begin{equation}k=w(m-d+1)-\eta(m-d+1).\end{equation}
Moreover, for any $a\in[0,wm-k]$ with $\tau(a)\geqslant m-d+1$, it holds that
\begin{equation}\tau(a)=\left\lceil\frac{a+k}{w}\right\rceil.\end{equation}

\end{proposition}

\begin{proof}
We only prove (1) and the proof of (2) is similar. First, for any $t\in[m+1-e,m]$, by $m-t\leqslant e-1$, we have $0=\eta(m-t)=\psi(t)+w(m-t)-k$, which establishes (3.7), as desired. In particular, we have $\psi(m+1-e)=k-w(e-1)$, which establishes (3.8), as desired. Now consider $a\in[0,k]$ with $\varphi(a)\geqslant m-e+1$. By (3.7), we have $a\leqslant\psi(\varphi(a))=k-w(m-\varphi(a))$. This implies that $m-\lfloor\frac{k-a}{w}\rfloor\leqslant\varphi(a)$, which, together with (3.6), further establishes (3.9), as desired.
\end{proof}

\subsection{The case that $\varphi(1)+\tau(1)=m+2$}
\setlength{\parindent}{2em}
In this subsection, we prove the following theorem, which can be used to characterize MDS and MRD codes.

\setlength{\parindent}{0em}
\begin{theorem}
Suppose that $1\leqslant k\leqslant wm-1$, $d=\varphi(1)$ and $e=\tau(1)$. Then, we have
\begin{equation}w(e-1)\leqslant k\leqslant w(m-d+1).\end{equation}
In particular, it holds that $d+e\leqslant m+2$. Moreover, the following three statements are equivalent to each other:

{\bf{(1)}}\,\,$k=w(m-d+1)$;

{\bf{(2)}}\,\,$k=w(e-1)$;

{\bf{(3)}}\,\,$d+e=m+2$.
\end{theorem}

\begin{proof}
First of all, (3.13) is a combination of (3.8) and (3.11), which further implies that $e-1\leqslant m-d+1$ and hence $d+e\leqslant m+2$, as desired. Moreover, (3.11) and (3.8) imply that
$$k=w(m-d+1)\Longleftrightarrow\eta(m-d+1)=0\Longleftrightarrow m-d+1\leqslant e-1\Longleftrightarrow m+2\leqslant d+e,$$
$$k=w(e-1)\Longleftrightarrow\psi(m-e+1)=0\Longleftrightarrow m-e+1\leqslant d-1\Longleftrightarrow m+2\leqslant d+e,$$
which, together with $d+e\leqslant m+2$, further establishes the desired result.
\end{proof}

\setlength{\parindent}{0em}
\begin{example}
{\bf{(1)}}\,\,Following (1) of Example 2.5 and consider a linear code $C\subseteq\mathbb{F}^{\Omega}$ with $\dim_{\mathbb{F}}(C)=k\in[1,m-1]$. Then, Theorem 3.4 implies that  $C$ is an MDS $\mathbf{P}$-code if and only if $C^{\bot}$ is an MDS $\overline{\mathbf{P}}$-code, if and only if $\mathbf{d}_1^{\mathbf{P}}(C)+\mathbf{d}_1^{\overline{\mathbf{P}}}(C^{\bot})=m+2$, which recovers [20, Theorem 3.12]. In particular, if $\mathbf{P}$ is an anti-chain, then Theorem 3.4 recovers the well-known fact that $C$ is MDS if and only if $C^{\bot}$ is MDS.

{\bf{(2)}}\,\,Following (2) of Example 2.5 and consider a Gabidulin code $C\subseteq\mathbb{E}^{[m]}$ with $\dim_{\mathbb{E}}(C)=k\in[1,m-1]$. Then, Theorem 3.4 implies that $C$ is MRD if and only if $C^{\bot}$ is MRD, if and only if $\mathbf{d}_1^{G}(C)+\mathbf{d}_1^{G}(C^{\bot})=m+2$.

{\bf{(3)}}\,\,Following (3) of Example 2.5 and consider a Delsarte code $C\subseteq\Mat_{m,n}(\mathbb{F})$ with $\dim_{\mathbb{F}}(C)=k\in[1,mn-1]$. Then, Theorem 3.4 implies that $C$ is MRD if and only if $C^{\bot}$ is MRD, if and only if $\mathbf{d}_1^{M}(C)+\mathbf{d}_1^{M}(C^{\bot})=m+2$, as has been shown in [8, Proposition 19 (1)].
\end{example}

\subsection{A duality theorem}

\setlength{\parindent}{2em}
In this subsection, we prove a duality theorem which will be used frequently in the rest of this section.

\setlength{\parindent}{0em}
\begin{theorem}
{\bf{(1)}}\,\,Let $c\in[0,wm-k]$ and $l\in\mathbb{Z}$. Then, $\tau(c)\leqslant l$ if and only if both $c\leqslant wl$ and
$$c\geqslant wl-k+1\Longrightarrow\varphi(c+k-wl)\leqslant m-l$$
hold true. Moreover, $\tau(c)\geqslant l$ if and only if both $c\geqslant w(l-1)-k+1$ and
$$c\leqslant w(l-1)\Longrightarrow\varphi(c-w(l-1)+k)\geqslant m-l+2$$
hold true.

{\bf{(2)}}\,\,Let $a\in[0,k]$ and $b\in\mathbb{Z}$. Then, $\varphi(a)\leqslant b$ if and only if both $a\leqslant wb$ and
$$a\geqslant k+1-w(m-b)\Longrightarrow\tau(a+w(m-b)-k)\leqslant m-b$$
hold true. Moreover, $\varphi(a)\geqslant b$ if and only if both $a\geqslant k+1-w(m-b+1)$ and
$$a\leqslant w(b-1)\Longrightarrow\tau(a+w(m-b+1)-k)\geqslant m-b+2$$
hold true.
\end{theorem}

\begin{proof}
{\bf{(1)}}\,\,We only prove the first assertion as the ``moreover'' part follows from an application of the first assertion to $(c,l-1)$. In what follows, we assume that $l\in[0,m]$ as otherwise the result immediately follows. First, note that $\tau(c)\leqslant l$ if and only if $c\leqslant \eta(l)$, if and only if $c+k-wl\leqslant\psi(m-l)$. Next, we prove the ``only if'' part. By $\tau(c)\leqslant l$, we have $c+k-wl\leqslant\psi(m-l)\leqslant k$ and hence $c\leqslant wl$, as desired. Moreover, if $c\geqslant wl-k+1$, then it follows from $c+k-wl\leqslant\psi(m-l)$ that $\varphi(c+k-wl)\leqslant m-l$, as desired. Finally, we prove the ``if'' part. If $c\leqslant wl-k$, then we have $c+k-wl\leqslant0\leqslant \psi(m-l)$, as desired. And if $c\geqslant wl-k+1$, then from $\varphi(c+k-wl)\leqslant m-l$, we deduce that $c+k-wl\leqslant \psi(m-l)$, as desired.

{\bf{(2)}}\,\,The proof of the first assertion is similar to that of (1), and the ``moreover'' assertion follows from an application of the first assertion to $(a,b-1)$. We omit the details of the verification.
\end{proof}

\setlength{\parindent}{2em}
Theorem 3.5 shows that $\varphi$ and $\tau$ are determined by each other. It is convenient to apply Theorem 3.5 in various cases as substituting $(c,l)$ or $(a,b)$ accordingly will lead to the desired result. Theorem 3.5 may be regarded as an analogue of the Wei-type duality theorem for $\varphi$ and $\tau$ (see Theorem 3.12) as many coding-theoretic results derived from Wei-type duality theorems can be established by a simple application of Theorem 3.5, as will be shown in later sections. In addition, Theorem 3.5 is very elementary in the sense that it directly follows from (3.2) and the definition of Galois connection.

\subsection{Necessary and sufficient conditions for obtaining the generalized Singleton bound}
\setlength{\parindent}{2em}
As a first application of Theorem 3.5, in this subsection, we derive necessary and sufficient conditions for equality to hold in (3.6). The following is the main result of this subsection.

\setlength{\parindent}{0em}
\begin{theorem}
{\bf{(1)}}\,\,Let $\gamma\in[1,w]$ with $k\geqslant w(m-1)+\gamma$. Then, for any $a\in[0,k]$ with $a\equiv\gamma~(\bmod~w)$, it holds that $\varphi(a)=m-\lfloor\frac{k-a}{w}\rfloor$.

{\bf{(2)}}\,\,Let $\gamma\in[1,w]$ with $k\leqslant w(m-1)+\gamma-1$. Then, for any $a\in[0,k]$ with $a\equiv\gamma~(\bmod~w)$, it holds that
$$\varphi(a)=m-\left\lfloor\frac{k-a}{w}\right\rfloor\Longleftrightarrow\tau\left(\left\lceil\frac{k-\gamma+1}{w}\right\rceil w-k+\gamma\right)\geqslant \left\lfloor\frac{k-a}{w}\right\rfloor+2.$$
\end{theorem}

\begin{proof}
With (3) of Theorem 3.2 and some straightforward computation, the desired result follows from an application of (2) of Theorem 3.5 to $a$ and $b=m-\lfloor\frac{k-a}{w}\rfloor$.
\end{proof}

\setlength{\parindent}{0em}
\begin{corollary}
Let $\gamma\in[1,w]$ with $k\leqslant w(m-1)+\gamma-1$, and let $\theta=\tau(\lceil\frac{k-\gamma+1}{w}\rceil w-k+\gamma)$. Then, for any $a\in[0,k]$ with $a\equiv\gamma~(\bmod~w)$, the following two statements are equivalent to each other:

{\bf{(1)}}\,\,$\theta=\left\lfloor\frac{k-a}{w}\right\rfloor+2$;

{\bf{(2)}}\,\,$\varphi(a)=m-\left\lfloor\frac{k-a}{w}\right\rfloor$, and $a=\min\{c\in[0,k]\mid c\equiv\gamma~(\bmod~w),\varphi(c)=m-\left\lfloor\frac{k-c}{w}\right\rfloor\}$.
\end{corollary}

\begin{proof}
$(1)\Longrightarrow(2)$\,\,By Theorem 3.6, we have $\varphi(a)=m-\lfloor\frac{k-a}{w}\rfloor$. Now let $c\in[0,k]$ such that $c\equiv\gamma~(\bmod~w)$ and $\varphi(c)=m-\lfloor\frac{k-c}{w}\rfloor$. By Theorem 3.6, we have $\theta\geqslant\lfloor\frac{k-c}{w}\rfloor+2$. It follows that $\lfloor\frac{k-a}{w}\rfloor\geqslant\lfloor\frac{k-c}{w}\rfloor$, which, together with $c\equiv a~(\bmod~w)$, further implies that $c\geqslant a$, as desired.

$(2)\Longrightarrow(1)$\,\,By Theorem 3.6, it remains to show that $\theta\leqslant\lfloor\frac{k-a}{w}\rfloor+2$. If $a\geqslant w$, then by (2), we have $\varphi(a-w)\neq m-\lfloor\frac{k-(a-w)}{w}\rfloor$, which, together with Theorem 3.6, implies that $\theta\leqslant\lfloor\frac{k-(a-w)}{w}\rfloor+1=\lfloor\frac{k-a}{w}\rfloor+2$, as desired. Moreover, if $a\leqslant w-1$, then we have $a\leqslant\gamma$, which, together with Theorem 3.3, further implies that $\theta\leqslant\lfloor\frac{k-\gamma}{w}\rfloor+2\leqslant\lfloor\frac{k-a}{w}\rfloor+2$, as desired.
\end{proof}

\setlength{\parindent}{2em}
Theorem 3.6 and Corollary 3.1 can be used to characterize $i$-MDS codes and $i$-MRD codes, as detailed in the following example.

\setlength{\parindent}{0em}
\begin{example}
{\bf{(1)}}\,\,Following (1) of Example 2.5 and consider a linear code $C\subseteq\mathbb{F}^{\Omega}$ with $\dim_{\mathbb{F}}(C)=k$. For any $i\in[0,k]$, we say that $C$ is an $i$-MDS $\mathbf{P}$-code if $\mathbf{d}_i^{\mathbf{P}}(C)=m-k+i$. Suppose that $k\leqslant m-1$ and $a\in[0,k]$. Then, Theorem 3.6 and Corollary 3.1 imply that $C$ is $a$-MDS if and only if $\mathbf{d}_1^{\overline{\mathbf{P}}}(C^{\bot})\geqslant k-a+2$; moreover, $\mathbf{d}_1^{\overline{\mathbf{P}}}(C^{\bot})=k-a+2$ if and only if $C$ is $a$-MDS and $a=\min\{c\in[0,k]\mid \text{$C$ is $c$-MDS}\}$. In particular, if $\mathbf{P}$ is an anti-chain, then the counterpart result has been established in [14, Proposition 4] by using the Wei-type duality theorem for GHWs (see [34, Theorem 3]).

{\bf{(2)}}\,\,Following (2) of Example 2.5 and consider a Gabidulin code $C\subseteq\mathbb{E}^{[m]}$ with $\dim_{\mathbb{E}}(C)=k$. For any $i\in[0,k]$, we say that $C$ is $i$-MRD if $\mathbf{d}_i^{G}(C)=m-k+i$ (see [13, Definition 1]). Suppose that $k\leqslant m-1$ and $a\in[0,k]$. Then, Theorem 3.6 implies that $C$ is $a$-MRD if and only if $\mathbf{d}_1^{G}(C^{\bot})\geqslant k-a+2$, as has been shown in [13, Corollary III.3] by using the Wei-type duality theorem for GRWs (see [13, Theorem I.3]).

{\bf{(3)}}\,\,Suppose that $\mathbb{F}$ is a field, $C\subseteq\Mat_{m,w}(\mathbb{F})$ is a Delsarte code with $\dim_{\mathbb{F}}(C)=k$, and $\varphi(r)$ is equal to the $r$-th Delsarte generalized weight of $C$ for all $r\in[0,k]$ (see [28, Definition 23]). Then, $\tau(s)$ is equal to the $s$-th Delsarte generalized weight of $C^{\bot}$ for all $s\in[0,wm-k]$. For any $i\in[1,\lceil\frac{k}{w}\rceil]$, we say that $C$ is $i$-MRD if $\varphi((i-1)w+1)=m-\lceil\frac{k}{w}\rceil+i$ (see [7, Definition 6.1]). Then, with $\gamma$ set to be $1$, Theorem 3.6 and Corollary 3.1 recover counterpart results in [7, Section 6.1], which were established for $i$-MRD codes by using the Wei-type duality theorem for Delsarte generalized weights (see [28, Corollary 38]).
\end{example}

\subsection{MDS discrepancy}

\setlength{\parindent}{2em}
Following [34, Section VI], the \textit{MDS discrepancy} of a $k$-dimensional linear code $C\subseteq\mathbb{F}^{m}$, where $\mathbb{F}$ is a field, is defined as the largest integer $s\in[0,k]$ satisfying $\mathbf{d}_s(C)\leqslant m-k$. MDS discrepancy measures how far the code $C$ is away from an MDS code. It has been proven in [34, Theorem 10] that the MDS discrepancy of a code is equal to that of its dual code, which has been further generalized to poset codes in \cite{27} by using the Wei-type duality theorem for poset codes (see [27, Theorem 2]). In terms of Galois connections, these results may be illustrated by the following proposition.

\setlength{\parindent}{0em}
\begin{proposition}
Suppose that $w\mid k$ and $1\leqslant\frac{k}{w}\leqslant m-1$. Then, for any $c\in[1,\min(k,wm-k)]$, it holds that $\tau(c)\leqslant \frac{k}{w}\Longleftrightarrow\varphi(c)\leqslant m-\frac{k}{w}$.
\end{proposition}

\begin{proof}
An application of (1) of Theorem 3.5 to $l=\frac{k}{w}$ immediately implies the desired result.
\end{proof}

\setlength{\parindent}{2em}
We apply Proposition 3.5 to poset codes and rank metric codes, as detailed in the following example.

\setlength{\parindent}{0em}
\begin{example}
{\bf{(1)}}\,\,Following (1) of Example 2.5 and consider a linear code $C\subseteq\mathbb{F}^{\Omega}$ with $\dim_{\mathbb{F}}(C)=k\in[1,m-1]$. For any $s\in[1,\min\{k,m-k\}]$, Proposition 3.5 implies that $\mathbf{d}_s^{\overline{\mathbf{P}}}(C^{\bot})\leqslant k\Longleftrightarrow\mathbf{d}_s^{\mathbf{P}}(C)\leqslant m-k$. Hence the $\mathbf{P}$-MDS discrepancy of $C$ is equal to the $\overline{\mathbf{P}}$-MDS discrepancy of $C^{\bot}$ (see [27, Section 4.1]). This recovers [27, Theorem 3], which further recovers [34, Theorem 10] when $\mathbf{P}$ is an anti-chain.

{\bf{(2)}}\,\,Following (2) of Example 2.5 and consider a Gabidulin code $C\subseteq\mathbb{E}^{[m]}$ with $\dim_{\mathbb{E}}(C)=k\in[1,m-1]$. For any $s\in[1,\min\{k,m-k\}]$, Proposition 3.5 implies that $\mathbf{d}_s^{G}(C^{\bot})\leqslant k\Longleftrightarrow\mathbf{d}_s^{G}(C)\leqslant m-k$. Hence if we define the MRD discrepancy of $C$ as the largest integer $s\in[0,k]$ satisfying $\mathbf{d}_{s}^{G}(C)\leqslant m-k$, then the MRD discrepancy of $C$ is equal to that of $C^{\bot}$.

{\bf{(3)}}\,\,Following (3) of Example 2.5 and consider a Delsarte code $C\subseteq\Mat_{m,n}(\mathbb{F})$ with $\dim_{\mathbb{F}}(C)=k\in[1,mn-1]$ and $n\mid k$. Then, for any $s\in[1,\min\{k,mn-k\}]$, Proposition 3.5 implies that $\mathbf{d}_s^{M}(C^{\bot})\leqslant \frac{k}{n}\Longleftrightarrow\mathbf{d}_s^{M}(C)\leqslant m-\frac{k}{n}$. Hence if we define the MRD discrepancy of $C$ as the largest integer $s\in[0,k]$ satisfying $\mathbf{d}_{s}^{M}(C)\leqslant m-\frac{k}{n}$, then the MRD discrepancy of $C$ is equal to that of $C^{\bot}$.
\end{example}

\subsection{The case that $\varphi(1)+\tau(1)=m$}
\setlength{\parindent}{2em}
In this subsection, we establish some results for the special case $\varphi(1)+\tau(1)=m$, which can be used to characterize near MDS (NMDS) and near MRD (NMRD) codes. We begin with the following proposition.

\setlength{\parindent}{0em}
\begin{proposition}
Suppose that $1\leqslant k\leqslant wm-1$, $d=\varphi(1)$, $e=\tau(1)$, $w\mid k$ and $d+e\neq m+2$. Then, it holds that
\begin{equation}e\leqslant\frac{k}{w}\leqslant m-d.\end{equation}
Consequently, we have $d+e\leqslant m$. Moreover, if $d+e=m$, then we have $k=we=w(m-d)$.
\end{proposition}

\begin{proof}
By $d+e\neq m+2$ and Theorem 3.4, we have $\frac{k}{w}\neq e-1$ and $\frac{k}{w}\neq m-d+1$, which, together with $e-1\leqslant\frac{k}{w}\leqslant m-d+1$ (see (3.13)) and $\frac{k}{w}\in\mathbb{Z}$, immediately establishes (3.14), which further yields the rest part.
\end{proof}

\begin{corollary}
Suppose that $1\leqslant k\leqslant wm-1$, $d=\varphi(1)$, $e=\tau(1)$ and $d+e=m$. Then, $(d,e)$ is equal to either $(m-\left\lceil\frac{k}{w}\right\rceil,\left\lceil\frac{k}{w}\right\rceil)$ or $(m-\left\lfloor\frac{k}{w}\right\rfloor,\left\lfloor\frac{k}{w}\right\rfloor)$.
\end{corollary}

\begin{proof}
By Proposition 3.6, we can assume that $w\nmid k$. From $w(e-1)\leqslant k\leqslant w(m-d+1)$, we deduce that $e\leqslant\left\lceil\frac{k}{w}\right\rceil$ and $d\leqslant m-\left\lceil\frac{k}{w}\right\rceil+1$. This, together with $(m-\left\lceil\frac{k}{w}\right\rceil+1)+\left\lceil\frac{k}{w}\right\rceil=m+1=d+e+1$, implies that either $e=\left\lceil\frac{k}{w}\right\rceil$ or $e=\left\lceil\frac{k}{w}\right\rceil-1=\left\lfloor\frac{k}{w}\right\rfloor$, which further establishes the desired result.
\end{proof}

\setlength{\parindent}{2em}
The following result follows from Proposition 3.6 and Theorem 3.5.

\setlength{\parindent}{0em}
\begin{corollary}
Suppose that $1\leqslant k\leqslant wm-1$, $d=\varphi(1)$ and $e=\tau(1)$. Then, the following two statements are equivalent to each other:

{\bf{(1)}}\,\,$w\mid k$ and $d+e=m$;

{\bf{(2)}}\,\,$k=w(m-d)$, and moreover, $d\leqslant m-2\Longrightarrow\varphi(w+1)=d+2$.
\end{corollary}

\begin{proof}
${\bf{(1)}}\Longrightarrow{\bf{(2)}}$\,\,By Proposition 3.6, we have $k=w(m-d)$. Moreover, if $d\leqslant m-2$, then an application of (1) of Theorem 3.5 to $c=1$ and $l=m-d$ implies that $\varphi(w+1)\geqslant d+2$, which, together with (3) of Theorem 3.2, further implies that $\varphi(w+1)=d+2$, as desired.

${\bf{(2)}}\Longrightarrow{\bf{(1)}}$\,\,By Theorem 3.4, we have $w\mid k$, $d+e\neq m+2$, which, together with Proposition 3.6, implies that $d+e\leqslant m$. Hence if $d\geqslant m-1$, then from $e\geqslant1$, we deduce that $d+e=m$, as desired. And if $d\leqslant m-2$, then we have $\varphi(w+1)=d+2$, and an application of (1) of Theorem 3.5 to $c=1$ and $l=m-d$ implies that $e=m-d$, as desired.
\end{proof}

\setlength{\parindent}{2em}
Next, we compute $\varphi$ and $\psi$ explicitly when $\varphi(1)+\tau(1)=m$.

\setlength{\parindent}{0em}
\begin{theorem}
Suppose that $1\leqslant k\leqslant wm-1$, $d=\varphi(1)$, $e=\tau(1)$ and $d+e=m$, and let $\theta=\max\{a\in[1,k]\mid\varphi(a)=d\}$. Then, we have $\psi(d)=\theta$, $\psi(t)=k-w(m-t)$ for all $t\in[d+1,m]$, $\varphi(c)=d$ for all $c\in[1,\theta]$, and $\varphi(a)=m-\left\lfloor\frac{k-a}{w}\right\rfloor$ for all $a\in[\theta+1,k]$.
\end{theorem}

\begin{proof}
We note that $\psi(d)=\max\{c\in [0,k]\mid \varphi(c)\leqslant d\}=\theta$, and for any $c\in[1,\theta]$, it follows from $\varphi(1)\leqslant\varphi(c)\leqslant\varphi(\theta)=d$ that $\varphi(c)=d$, as desired. Moreover, noticing that $\varphi(a)\geqslant d+1=m+1-e$ for all $a\in[\theta+1,k]$, the rest immediately follows from (1) of Proposition 3.4.
\end{proof}

\setlength{\parindent}{2em}
Now we further improve Theorem 3.7 with the additional assumption that $w\mid k$.

\setlength{\parindent}{0em}
\begin{theorem}
Suppose that $1\leqslant k\leqslant wm-1$, $w\mid k$, $d=\varphi(1)$, $e=\tau(1)$ and $d+e=m$, and let $\theta=\max\{a\in[1,k]\mid\varphi(a)=d\}$. Then, the following four statements hold:

{\bf{(1)}}\,\,$\psi(d)=\theta$. Moreover, we have $\psi(t)=w(t-d)$ for all $t\in[d+1,m]$;

{\bf{(2)}}\,\,We have $\varphi(c)=d$ for all $c\in[1,\theta]$, and $\varphi(a)=d+\left\lceil\frac{a}{w}\right\rceil$ for all $a\in[\theta+1,k]$;

{\bf{(3)}}\,\,$\eta(e)=\theta$. Moreover, we have $\eta(l)=w(l-e)$ for all $l\in[e+1,m]$;

{\bf{(4)}}\,\,We have $\tau(c)=e$ for all $c\in[1,\theta]$, and $\tau(a)=\left\lceil\frac{a}{w}\right\rceil+e$ for all $a\in[\theta+1,wm-k]$.
\end{theorem}

\begin{proof}
By Proposition 3.6, we have $k=w(m-d)=we$. Hence (1) and (2) follow from Theorem 3.7. Next, we prove (3). By (1), we have $\eta(e)=\psi(m-e)+we-k=\psi(d)=\theta$; moreover, the rest follows from (3.10) (see Proposition 3.4). Finally, we prove (4). For any $c\in[1,\theta]$, by (3), we have $c\leqslant\eta(e)$ and hence $\tau(c)\leqslant e$, which, together with $\tau(c)\geqslant\tau(1)=e$, implies that $\tau(c)=e$, as desired. Moreover, for any $a\in[\theta+1,wm-k]$, by $a\geqslant\eta(e)+1$, we have $\tau(a)\geqslant e+1=m-d+1$, which, together with (3.12) (see Proposition 3.4), further establishes the desired result.
\end{proof}

\setlength{\parindent}{2em}
We end this subsection by recovering some characterizations of NMDS and NMRD codes, as detailed in the following example.

\setlength{\parindent}{0em}
\begin{example}
{\bf{(1)}}\,\,Following (1) of Example 2.5 and consider a linear code $C\subseteq\mathbb{F}^{\Omega}$ with $\dim_{\mathbb{F}}(C)=k\in[1,m-1]$. Following [2, Definition 2.3], we say that $C$ is a near-MDS (NMDS) $\mathbf{P}$-code if $\mathbf{d}_1^{\mathbf{P}}(C)=m-k$, and $k\geqslant2\Longrightarrow\mathbf{d}_2^{\mathbf{P}}(C)=m-k+2$. Corollary 3.3 implies that $C$ is NMDS if and only if $\mathbf{d}_1^{\mathbf{P}}(C)+\mathbf{d}_1^{\overline{\mathbf{P}}}(C^{\bot})=m$, as has been shown in [2, Lemma 2.4]. In particular, if $\mathbf{P}$ is an anti-chain, then Corollary 3.3 recovers the well-known fact that $C$ is NMDS if and only if $\mathbf{d}_1(C)+\mathbf{d}_1(C^{\bot})=m$ (see [14, Definition 13]).

{\bf{(2)}}\,\,Following (2) of Example 2.5 and consider a Gabidulin code $C\subseteq\mathbb{E}^{[m]}$ with $\dim_{\mathbb{E}}(C)=k\in[1,m-1]$. Following [25, Definition 5.1], we say that $C$ is near-MRD (NMRD) if $\mathbf{d}_1^{G}(C)=m-k$, and $k\geqslant2\Longrightarrow\mathbf{d}_2^{G}(C)=m-k+2$. Corollary 3.3 implies that $C$ is NMRD if and only if $\mathbf{d}_1^{G}(C)+\mathbf{d}_1^{G}(C^{\bot})=m$, as has been shown in [25, Proposition 5.3].

{\bf{(3)}}\,\,Following (3) of Example 3.6 and consider a Delsarte code $C\subseteq\Mat_{m,w}(\mathbb{F})$ with $\dim_{\mathbb{F}}(C)=k\in[1,wm-1]$ and $w\mid k$. Following [7, Definition 6.8], we say that $C$ is near-MRD (NMRD) if $k=w(m-d)$, and moreover, $k\geqslant2w$ implies that the $(w+1)$-th Delsarte generalized weight of $C$ is equal to $d+2$. Corollary 3.3 implies that $C$ is NMRD if and only if $\mathbf{d}_1^{M}(C)+\mathbf{d}_1^{M}(C^{\bot})=m$, as has been shown in [7, Theorem 6.9]. In addition, Corollary 3.2 recovers [8, Remark 36].
\end{example}

\subsection{The case that $\varphi(1)+\tau(1)=m+1$}
\setlength{\parindent}{2em}
In this subsection, we establish some results for the special case $\varphi(1)+\tau(1)=m+1$, which can be used to characterize dually quasi-MRD (dually QMRD) codes. We begin with the following theorem.

\setlength{\parindent}{0em}
\begin{theorem}
Suppose that $1\leqslant k\leqslant wm-1$, $d=\varphi(1)$ and $e=\tau(1)$. Then, the following three statements are equivalent to each other:

{\bf{(1)}}\,\,$d+e=m+1$;

{\bf{(2)}}\,\,$w\nmid k$, $d=m-\left\lceil\frac{k}{w}\right\rceil+1$, $e=\left\lceil\frac{k}{w}\right\rceil$;

{\bf{(3)}}\,\,$\left\lceil\frac{k}{w}\right\rceil=m-d+1=e$.
\end{theorem}

\begin{proof}
It suffices to prove $(1)\Longrightarrow(2)$. From $d+e\neq m+2$ and Proposition 3.6, we deduce that $w\nmid k$, as desired. Moreover, from $w(e-1)\leqslant k\leqslant w(m-d+1)$, we deduce that $e\leqslant\lceil\frac{k}{w}\rceil$ and $d\leqslant m+1-\lceil\frac{k}{w}\rceil$, which, together with $d+e=m+1$, immediately establishes (2), as desired.
\end{proof}

\setlength{\parindent}{2em}
Next, we compute $\varphi$, $\psi$, $\tau$ and $\eta$ explicitly when $\varphi(1)+\tau(1)=m+1$. The following theorem is an immediate corollary of Proposition 3.4.

\begin{theorem}
Suppose that $1\leqslant k\leqslant wm-1$, $d=\varphi(1)$, $e=\tau(1)$ and $d+e\geqslant m+1$. Then, we have $\psi(t)=k-w(m-t)$ for all $t\in[d,m]$, $\eta(l)=wl-k$ for all $l\in[e,m]$, $\varphi(a)=m-\left\lfloor\frac{k-a}{w}\right\rfloor$ for all $a\in[1,k]$, and $\tau(c)=\left\lceil\frac{c+k}{w}\right\rceil$ for all $c\in[1,wm-k]$.
\end{theorem}

\setlength{\parindent}{2em}
Now we give the following necessary and sufficient conditions for $\varphi(1)+\tau(1)\geqslant m+1$.

\setlength{\parindent}{0em}
\begin{proposition}
Suppose that $1\leqslant k\leqslant wm-1$, $d=\varphi(1)$ and $e=\tau(1)$. Then, the following two statements are equivalent to each other:

{\bf{(1)}}\,\,$d+e\geqslant m+1$;

{\bf{(2)}}\,\,$k\geqslant w(m-d)+1$. Moreover, $d\leqslant m-1\Longrightarrow\varphi(k-w(m-d)+1)=d+1$.
\end{proposition}

\begin{proof}
$(1)\Longrightarrow(2)$\,\,It follows from Theorems 3.4 and 3.9 that $k\geqslant w(m-d)+1$; moreover, the rest follows from Theorem 3.10.

$(2)\Longrightarrow(1)$\,\,If $d=m$, then it follows from $e\geqslant1$ that $d+e\geqslant m+1$, as desired. Moreover, if $d\leqslant m-1$, then an application of (2) of Theorem 3.5 to $a=k-w(m-d)+1$ and $b=d+1$ implies that $e\geqslant m-d+1$, as desired.
\end{proof}

\setlength{\parindent}{2em}
Finally, we use results established in this subsection to recover some characterizations of dually QMRD codes, which has been proposed in \cite{8}.

\setlength{\parindent}{0em}
\begin{example}
Following (3) of Example 2.5 and consider a Delsarte code $C\subseteq\Mat_{m,n}(\mathbb{F})$ with $\dim_{\mathbb{F}}(C)=k\in[1,mn-1]$. By [8, Definitions 10 and 13], we say that $C$ is dually QMRD if $n\nmid k$, $\mathbf{d}_1^{M}(C)=m-\left\lceil\frac{k}{n}\right\rceil+1$ and $\mathbf{d}_1^{M}(C^{\bot})=\left\lceil\frac{k}{n}\right\rceil$. Theorem 3.9 implies that $C$ is dually QMRD if and only if $\mathbf{d}_1^{M}(C)+\mathbf{d}_1^{M}(C^{\bot})=m+1$, as has been shown in [8, Proposition 19]. Theorem 3.10 implies that if $C$ is either MRD or dually QMRD, then the GMWs of $C$ coincide with the Delsarte generalized weights of $C$, which are completely determined by $m,n,k$, as has been shown in [8, Theorem 22 and Remark 23]. Moreover, Proposition 3.7 gives necessary and sufficient conditions for $C$ to be either MRD or dually QMRD (cf. [8, Theorem 22]).
\end{example}

\subsection{Evasive property}
\setlength{\parindent}{2em}
We begin by recalling some notions on evasive property (see \cite{3,25}). Let $\mathbb{E}/\mathbb{F}$ be a finite dimensional field extension. Following [3, Definition 1.1], for $(h,r)\in\mathbb{N}^{2}$, an $\mathbb{F}$-subspace $U\subseteq\mathbb{E}^{k}$ is said to be \textit{$(h,r)$-evasive} in $\mathbb{E}^{k}$ if $\dim_{\mathbb{F}}(U\cap W)\leqslant r$ for all $\mathbb{E}$-subspace $W\subseteq\mathbb{E}^{k}$ with $\dim_{\mathbb{E}}(W)=h$.

$(h,r)$-evasiveness is closely related to GRWs of Gabidulin codes. More precisely, let $C\subseteq\mathbb{E}^{[m]}$ be a Gabidulin code such that $\dim_{\mathbb{E}}(C)=k$ and $\mathbf{d}_k^{G}(C)=m$, $L\in\Mat_{m,k}(\mathbb{E})$ be a generator matrix of $C$, i.e., $C=\{L\gamma^{T}\mid\gamma\in\mathbb{E}^{k}\}$, and let $U$ be the $\mathbb{F}$-subspace of $\mathbb{E}^{k}$ generated by all the rows of $L$. Then, for $(h,r)\in\mathbb{N}^{2}$ with $h\leqslant \min\{k-1,r\}$, it has been proven in [25, Theorem 3.3] that the following three statements are equivalent to each other:

$(i)$\,\,$U$ is $(h,r)$-evasive in $\mathbb{E}^{k}$;

$(ii)$\,\,$\mathbf{d}_{k-h}^{G}(C)\geqslant m-r$;

$(iii)$\,\,$\mathbf{d}_{r-h+1}^{G}(C^{\bot})\geqslant r+2$.\vspace*{2mm}\\
Noticing that if we follow Example 2.3 and set $\varphi$ to be the GRWs of $C$, then Condition $(ii)$ is equivalent to $\varphi(k-h)\geqslant m-r$. Based on this observation, we propose the following definition of $(h,r)$-evasiveness for $\varphi$ in the general case.

\begin{definition}
For $(h,r)\in\mathbb{N}^{2}$, we say that $\varphi$ is \textit{$(h,r)$-evasive} provided that $k\geqslant w(h+1)$ and $\varphi(k-w(h+1)+1)\geqslant m-r$.
\end{definition}

We collect some basic properties of $(h,r)$-evasiveness in the following proposition.

\setlength{\parindent}{0em}
\begin{proposition}
{\bf{(1)}}\,\,Let $h\in\mathbb{N}$ with $k\geqslant w(h+1)$. Then, we have $\varphi(k-w(h+1)+1)\leqslant m-h$. Hence $\varphi$ is $(h,h)$-evasive if and only if $\varphi(k-w(h+1)+1)=m-h$.

{\bf{(2)}}\,\,Let $(h,r)\in\mathbb{N}^{2}$ such that $\varphi$ is $(h,r)$-evasive. Then, we have $h\leqslant r$, and moreover, $\varphi$ is $(h-a,r-a)$-evasive for all $a\in[0,h]$.
\end{proposition}

\begin{proof}
{\bf{(1)}}\,\,This follows from the generalized Singleton bound (3.6).

{\bf{(2)}}\,\,By (1), we have $m-r\leqslant\varphi(k-w(h+1)+1)\leqslant m-h$, which implies that $h\leqslant r$, as desired. Moreover, for any $a\in[0,h]$, by Theorem 3.2, we have $\varphi(k-w(h-a+1)+1)\geqslant\varphi(k-w(h+1)+1)+a\geqslant m-r+a=m-(r-a)$, which implies that $\varphi$ is $(h-a,r-a)$-evasive, as desired.
\end{proof}

\setlength{\parindent}{2em}
The following theorem is the main result of this subsection, in which we characterize $(h,r)$-evasiveness of $\varphi$ in terms of $\tau$.

\setlength{\parindent}{0em}
\begin{theorem}
Let $(h,r)\in\mathbb{N}^{2}$ with $h\leqslant r\leqslant m-1$ and $w(h+1)\leqslant k\leqslant w(m+h-r)-1$. Then, $\varphi$ is $(h,r)$-evasive if and only if $\tau(w(r-h)+1)\geqslant r+2$.
\end{theorem}

\begin{proof}
An application of (2) of Theorem 3.5 to $a=k-w(h+1)+1$ and $b=m-r$ immediately implies the desired result.
\end{proof}

\setlength{\parindent}{2em}
Theorem 3.11 immediately implies the following characterization for $(h,h)$-evasiveness of $\varphi$.

\begin{corollary}
Let $h\in[0,m-2]$ with $w(h+1)\leqslant k\leqslant wm-1$, and let $e=\tau(1)$. Then, $\varphi$ is $(h,h)$-evasive if and only if $e\geqslant h+2$.
\end{corollary}

\setlength{\parindent}{0em}
\begin{remark}
With (2) of Example 2.5 and $(i)\Longleftrightarrow(ii)$ of [25, Theorem 3.3], Proposition 3.8 recovers [3, Proposition 2.6] and [25, Remark 3.2], and Theorem 3.11 recovers $(ii)\Longleftrightarrow(iii)$ of [25, Theorem 3.3], which has been established by using the Wei-type duality theorem for GRWs.
\end{remark}

\subsection{Wei-type duality theorem}
\setlength{\parindent}{2em}
We end this section with the following Wei-type duality theorem for $\varphi$ and $\tau$.

\begin{theorem}
For any $\gamma\in\mathds{Z}$, define $\mathcal{A}_{\gamma}=\{\varphi(u)\mid u\in[1,k],~u\equiv\gamma+k~(\bmod~w)\}$, $\mathcal{B}_{\gamma}=\{m+1-\tau(v)\mid v\in[1,wm-k],~v\equiv\gamma~(\bmod~w)\}$. Then, for any $\gamma\in\mathds{Z}$, it holds that $\mathcal{A}_{\gamma}\cap \mathcal{B}_{\gamma}=\emptyset$, $\mathcal{A}_{\gamma}\cup \mathcal{B}_{\gamma}=[1,m]$.
\end{theorem}

Theorem 3.12 is part of [35, Theorem 3.1]. It can be applied to various coding-theoretic scenarios to show that the generalized weights of a code and those of its dual code are determined by each other. For example, with Examples 2.5 and 3.6, Theorem 3.12 recovers the Wei-type duality theorems for GHWs \cite{34}, generalized poset weights \cite{2,27}, GRWs \cite{13}, GMWs \cite{26} and Delsarte generalized weights \cite{26}. In Sections 4 and 5, we will use Theorem 3.12 to derive two new Wei-type duality theorems for codes over modules endowed with rank metric and extended generalized poset weight. We refer the reader to \cite{35} for the proof and more applications of Theorem 3.12.

\section{Generalized rank weight for modules over chain rings}
Recently, there has been a growing interest in rank metric codes over rings (see \cite{4,22}). For example, in the very recent paper \cite{4}, the authors establish a MacWilliams identity for rank metric codes over finite chain rings (see [4, Theorem 3.21]), which is further used to show that the dual code of an MRD code is again MRD, generalizing well-known results for rank metric codes over finite fields (see [18, Theorem 26]). In this section, we propose and study generalized rank weights (GRWs) and profiles for codes over principal ideal rings. This two quantities in fact form a Galois connection. Hence many results of Section 3 can be applied. In Section 4.1, we establish some basic properties for GRWs and profiles for codes over principal left ideal rings, and in sections 4.2-4.7, we will focus on codes over chain rings. All the rings considered in this section can be infinite and non-commutative.

\subsection{Generalized rank weights for modules over principal ideal rings}
\setlength{\parindent}{2em}
Throughout this subsection, let $R$ be a principal left ideal ring, $M$ be a finitely generated left $R$-module with $\rank_R(M)=m\geqslant1$, and let $\Omega$ be a non-empty finite set with $|\Omega|=n$. For any $\alpha\in M^{\Omega}$, the \textit{rank support of $\alpha$} is defined as
$$\sigma(\alpha)\triangleq\sum_{i\in\Omega} R\cdot\alpha_i.$$
Moreover, for any $D\subseteq M^{\Omega}$, the \textit{rank support of $D$} is defined as
$$\Supp(D)\triangleq\sum_{\alpha\in D}\sigma(\alpha).$$

\begin{remark}
If $R$ is a field, $M=R^{[m]}$ and $\Omega=[1,n]$, then we can regard each $\alpha\in M^{\Omega}$ as a matrix in $\Mat_{m,n}(R)$, whose $i$-th column is equal to $\alpha_i$. With this identification, we have $\sigma(\alpha)=\col(\alpha)$ for all $\alpha\in M^{\Omega}$ and $\Supp(D)=\CSupp(D)$ for all $D\subseteq M^{\Omega}$ (see Example 2.4).
\end{remark}

Now let $C\leqslant_{R}M^{\Omega}$ with $\rank_R(C)=k$. Define $\varphi:[0,k]\longrightarrow\mathbb{N}$ and $\psi:[0,m]\longrightarrow\mathbb{N}$ as
\begin{equation}\varphi(r)=\min\left\{\rank_R(\Supp(D))\mid\text{$D\leqslant_{R}C$, $\rank_R(D)=r$}\right\},\end{equation}
\begin{equation}\psi(l)=\max\left\{\rank_R(C\cap U^{\Omega})\mid U\leqslant_{R}M,\rank_R(U)=l\right\}.\end{equation}

\begin{definition}
For any $r\in[0,k]$, $\varphi(r)$ is referred to as \textit{the $r$-th generalized rank weight (GRW) of $C$}, and for any $l\in[0,m]$, $\psi(l)$ is referred to as \textit{the $l$-th rank profile of $C$}. Moreover, if $C\neq\{0\}$, then $\varphi(1)$ is referred to as the \textit{minimal rank weight of $C$}.
\end{definition}

\begin{remark}
With Remark 4.1, Definition 4.1 recovers the notion of GMWs and DRPs for Delsarte codes proposed in [26, Definitions 10 and 11], and also recovers the notion of Delsarte generalized weights for Delsarte codes proposed in [28, Definition 23] when $m<n$.
\end{remark}

\setlength{\parindent}{2em}
Now we establish two basic properties of RGWs and profiles.

\setlength{\parindent}{0em}
\begin{proposition}
For any $r\in[0,k]$, it holds that
\begin{equation}\varphi(r)=\min\left\{\rank_R(U)\mid U\leqslant_{R}M,\rank_R(C\cap U^{\Omega})\geqslant r\right\}.\end{equation}
\end{proposition}

\begin{proof}
Let $r\in[0,k]$, and let $s$ denote the right hand side of (4.3). For any $D\leqslant_{R}C$ with $\rank_R(D)=r$, we have $\Supp(D)\leqslant_{R}M$ and $D\subseteq C\cap\Supp(D)^{\Omega}$, which, together with Lemma 2.4, implies that $\rank_R(C\cap \Supp(D)^{\Omega})\geqslant r$. This yields that $s\leqslant\varphi(r)$. Conversely, for any $U\leqslant_{R}M$ with $\rank_R(C\cap U^{\Omega})\geqslant r$, we can choose $D\leqslant_{R}C\cap U^{\Omega}$ with $\rank_R(D)=r$; from $D\subseteq U^{\Omega}$, we deduce that $\Supp(D)\subseteq U$, which, together with Lemma 2.4, implies that $\rank_R(\Supp(D))\leqslant\rank_R(U)$. This yields that $\varphi(r)\leqslant s$, which further establishes (4.3), as desired.
\end{proof}

\setlength{\parindent}{0em}
\begin{proposition}
For any $l\in[0,m]$, it holds that
\begin{equation}\psi(l)=\max\left\{\rank_R(C\cap U^{\Omega})\mid U\leqslant_{R}M,\rank_R(U)\leqslant l\right\}.\end{equation}
\end{proposition}

\begin{proof}
Let $l\in[0,m]$, and let $s$ denote the right hand side of (4.4). By definition, we have $\psi(l)\leqslant s$. Conversely, for any $U\leqslant_{R}M$ with $\rank_R(U)\leqslant l$, we can choose $V\leqslant_{R}M$ such that $\rank_R(V)=l$ and $U\subseteq V$; moreover, by Lemma 2.4, we have $\rank_R(C\cap U^{\Omega})\leqslant\rank_R(C\cap V^{\Omega})\leqslant\psi(l)$. This yields that $s\leqslant\psi(l)$, which further establishes (4.4), as desired.
\end{proof}

\setlength{\parindent}{2em}
Now we are ready to show that GRWs and profiles of $C$ form a Galois connection. The following theorem is the main result of this subsection.

\setlength{\parindent}{0em}
\begin{theorem}
We have $\varphi:[0,k]\longrightarrow[0,m]$ and $\psi:[0,m]\longrightarrow[0,k]$, and $(\varphi,\psi)$ is a Galois connection between $[0,k]$ and $[0,m]$ with $\psi(0)=0$. Moreover, for any $l\in[1,m]$, it holds that
\begin{equation}\psi(l)-\psi(l-1)\leqslant n,\end{equation}
and for any $r\in[0,k-n]$, it holds that
\begin{equation}\varphi(r)+1\leqslant\varphi(r+n).\end{equation}
\end{theorem}

\begin{proof}
First of all, it follows from Lemma 2.4 that $\ran(\varphi)\subseteq[0,m]$ and $\ran(\psi)\subseteq[0,k]$. Moreover, an application of Propositions 4.1, 4.2 and Lemma 2.1 implies that $(\varphi,\psi)$ is a Galois connection between $[0,k]$ and $[0,m]$ with $\psi(0)=\rank_R(C\cap \{0\}^{\Omega})=0$, as desired.

\hspace*{4mm}\,\,Now let $l\in[1,m]$. Then, we can choose $V\leqslant_{R}M$ such that $\rank_R(V)=l$ and $\rank_R(C\cap V^{\Omega})=\psi(l)$; moreover, we can choose $U\leqslant_{R}V$ such that $\rank_R(U)=l-1$ and $\rank_R(V/U)=1$. From Lemma 2.4, we deduce that
\begin{eqnarray*}
\begin{split}
\psi(l)-\psi(l-1)&\leqslant\rank_R(C\cap V^{\Omega})-\rank_R(C\cap U^{\Omega})\\
&\leqslant\rank_R((C\cap V^{\Omega})/(C\cap U^{\Omega}))\\
&\leqslant\rank_R(V^{\Omega}/U^{\Omega})\\
&=\rank_R((V/U)^{\Omega})\\
&\leqslant n,
\end{split}
\end{eqnarray*}
which establishes (4.5). Finally, (4.6) follows from (4.5) and Theorem 3.2, as desired.
\end{proof}

\setlength{\parindent}{2em}
As corollaries of Theorem 4.1, we derive Singleton bound and generalized Singleton bound for GRWs.

\setlength{\parindent}{0em}
\begin{corollary}
Suppose that $C\neq\{0\}$ and $d=\varphi(1)$. Then, we have
\begin{equation}d=\min\left\{\rank_R(U)\mid U\leqslant_{R}M,C\cap U^{\Omega}\neq\{0\}\right\}=\min\left\{\rank_R(\sigma(\alpha))\mid\alpha\in C,\alpha\neq0\right\}.\end{equation}
Moreover, we have $k\leqslant \rank_R(M^{n-d+1})\leqslant m(n-d+1)$ and $k\leqslant n(m-d+1)$.
\end{corollary}

\begin{proof}
The first equality of (4.7) follows from Proposition 4.1. Next, let $b$ denote the right hand side of (4.7). For any $D\leqslant_{R}C$ with $\rank_R(D)=1$, we can choose $\alpha\in D-\{0\}$; moreover, it follows from Lemma 2.4 that $\rank_R(\sigma(\alpha))\leqslant\rank_R(\Supp(D))$, which further implies that $b\leqslant d$. Conversely, for any $\alpha\in C-\{0\}$, we can choose $D\leqslant_{R}R\alpha$ with $\rank_R(D)=1$; moreover, it follows from $\Supp(D)\subseteq\Supp(R\alpha)=\sigma(\alpha)$ and Lemma 2.4 that $\rank_R(\Supp(D))\leqslant\rank_R(\sigma(\alpha))$, which further implies that $d\leqslant b$, as desired.

\hspace*{2mm}\,\,Now let $I\subseteq\Omega$ with $|I|=d-1$, and let $L=\{\alpha\in M^{\Omega}\mid\forall~i\in\Omega-I:\alpha_i=0\}$. For any $\alpha\in C\cap L$, by $\rank_R(\sigma(\alpha))\leqslant d-1$ and (4.7), we have $\alpha=0$. From $C\cap L=\{0\}$ and Lemma 2.4, we deduce that $\rank_R(C)\leqslant\rank_R(M^{\Omega}/L)=\rank_R(M^{n-d+1})\leqslant m(n-d+1)$, as desired. Finally, it follows from Theorem 4.1 and Proposition 3.3 that $k\leqslant n(m-d+1)$, as desired.
\end{proof}

\begin{corollary}
We have $\varphi(r)\leqslant m-\lfloor\frac{k-r}{n}\rfloor$ for all $r\in[0,k]$. Moreover, suppose that $a\in[0,k]$ and $\varphi(a)=m-\lfloor\frac{k-a}{n}\rfloor$. Then, for any $c\in[a,k]$ with $c\equiv a~(\bmod~n)$, it holds that $\varphi(c)=m-\lfloor\frac{k-c}{n}\rfloor$.
\end{corollary}

\begin{proof}
This follows from Theorems 4.1 and 3.2.
\end{proof}

\setlength{\parindent}{2em}
Now we define MRD/QMRD/NMRD/$i$-MRD codes and $(h,r)$-evasiveness of codes.

\setlength{\parindent}{0em}
\begin{definition}
Suppose that $C\neq\{0\}$ and $d$ is the minimal rank weight of $C$. Then,

{\bf{(1)}}\,\,We say that $C$ is MRD if $k=n(m-d+1)$;

{\bf{(2)}}\,\,We say that $C$ is quasi-MRD or QMRD if $n\nmid k$ and $\lceil\frac{k}{n}\rceil=m-d+1$;

{\bf{(3)}}\,\,We say that $C$ is near-MRD or NMRD provided that $k=n(m-d)$, and moreover, when $d\leqslant m-2$, the $(n+1)$-th GRW of $C$ is equal to $d+2$;

{\bf{(4)}}\,\,For any $i\in[1,\lceil\frac{k}{n}\rceil]$, we say that $C$ is $i$-MRD if the $((i-1)n+1)$-th GRW of $C$ is equal to $m-\lceil\frac{k}{n}\rceil+i$.
\end{definition}

\begin{definition}
For any $(h,r)\in\mathbb{N}^{2}$, we say that $C$ is $(h,r)$-evasive if $k\geqslant n(h+1)$ and the $(k-n(h+1)+1)$-th GRW of $C$ is larger than or equal to $m-r$.
\end{definition}

\setlength{\parindent}{2em}
Definitions 4.2 and 4.3 extend some notions for rank metric codes over fields, as detailed in the following remark.

\setlength{\parindent}{0em}
\begin{remark}
{\bf{(1)}}\,\,Following Remark 4.1, suppose that $R$ is a field, $M=R^{[m]}$ and $\Omega=[1,n]$. Then, (1) of Definition 4.2 recovers [17, Definition 3.8] when $m\leqslant n$, (2) of Definition 4.2 recovers [8, Definition 10], (3) of Definition 4.2 forms a subclass of the NMRD codes proposed in [7, Definition 6.8], and (4) of Definition 4.2 extends [7, Definition 6.1].

{\bf{(2)}}\,\,Following Example 3.2, let $\mathbb{E}/\mathbb{F}$ be a finite dimensional field extension with $\dim_{\mathbb{F}}(\mathbb{E})=n$, and set $R=\mathbb{F}$, $M=R^{[m]}$ and $\Omega=[1,n]$. Moreover, let $A\subseteq\mathbb{E}^{[m]}$ be a Gabidulin code with $\dim_{\mathbb{E}}(A)=p$ and $\mathbf{d}_p^{G}(A)=m$, $L\in\Mat_{m,p}(\mathbb{E})$ satisfy that $A=\{L\gamma^{T}\mid\gamma\in\mathbb{E}^{p}\}$, and let $U$ be the $\mathbb{F}$-subspace of $\mathbb{E}^{p}$ generated by all the rows of $L$. Suppose that $C$ is the Delsarte code associated to $A$ with respect of an ordered basis of $\mathbb{E}/\mathbb{F}$. For any $(h,r)\in\mathbb{N}^{2}$ with $h\leqslant \min\{p-1,r\}$, from [25, Theorem 3.3] and $\mathbf{d}_{(k-n(h+1)+1)}^{M}(C)=\mathbf{d}_{p-h}^{G}(A)$ (see Example 3.2), we deduce that $U$ is $(h,r)$-evasive in $\mathbb{E}^{p}$ if and only if $C$ is $(h,r)$-evasive.
\end{remark}

\setlength{\parindent}{2em}
We collect some connections among codes proposed in Definition 4.2 in the following lemma.

\setlength{\parindent}{0em}
\begin{lemma}
Suppose that $C\neq\{0\}$ and $d=\varphi(1)$. Then, the following three statements hold:

{\bf{(1)}}\,\,$C$ is $1$-MRD if and only if $\lceil\frac{k}{n}\rceil=m-d+1$, if and only if $C$ is either MRD or QMRD;

{\bf{(2)}}\,\,If $k\geqslant n(m-1)+1$, then we have $\lceil\frac{k}{n}\rceil=m$, $d=1$ and $C$ is $1$-MRD;

{\bf{(3)}}\,\,If $d\leqslant m-2$, then $C$ is NMRD if and only if $k=n(m-d)$ and $C$ is $2$-MRD.
\end{lemma}

\setlength{\parindent}{2em}
Now with the help of Theorem 4.1 and results established in Section 3, we give some basic properties of MRD codes, $i$-MRD codes and evasive property. First, we show that GRWs and profiles of MRD codes are completely determined by $m,n,k$. The following proposition is an immediate corollary of Theorems 4.1 and 3.1.

\setlength{\parindent}{0em}
\begin{proposition}
Suppose that $C\neq\{0\}$ is MRD and $d=\varphi(1)$. Then, we have $\psi(l)=n(l-d+1)$ for all $l\in[d-1,m]$, and $\varphi(a)=\left\lceil\frac{a}{n}\right\rceil+d-1$ for all $a\in[1,k]$.
\end{proposition}

\setlength{\parindent}{2em}
The following proposition is a corollary of the generalized Singleton bound (see Corollary 4.2), and it extends [7, Lemma 6.2] to $i$-MRD codes over rings.

\setlength{\parindent}{0em}
\begin{proposition}
Suppose that $k\geqslant1$. Then, the following three statements hold true:

{\bf{(1)}}\,\,For any $i\in[1,\lceil\frac{k}{n}\rceil]$, it holds that $\varphi((i-1)n+1)\leqslant m-\lceil\frac{k}{n}\rceil+i$;

{\bf{(2)}}\,\,Suppose that $i\in[1,\lceil\frac{k}{n}\rceil]$ and $C$ is $i$-MRD. Then, $C$ is $j$-MRD for all $j\in[i,\lceil\frac{k}{n}\rceil]$;

{\bf{(3)}}\,\,$C$ is $\lceil\frac{k}{n}\rceil$-MRD if and only if $\varphi(\lfloor\frac{k-1}{n}\rfloor n+1)=m$,  if and only if $C$ is $i$-MRD for some $i\in[1,\lceil\frac{k}{n}\rceil]$.
\end{proposition}

\setlength{\parindent}{2em}
The following proposition follows from Theorem 4.1 and Proposition 3.8. With (2) of remark 4.3, it generalizes [3, Proposition 2.6] and [25, Remark 3.2] to codes over rings.

\setlength{\parindent}{0em}
\begin{proposition}
{\bf{(1)}}\,\,Let $h\in\mathbb{N}$ with $k\geqslant n(h+1)$. Then, it holds that $\varphi(k-n(h+1)+1)\leqslant m-h$. Moreover, $C$ is $(h,h)$-evasive if and only if $\varphi(k-n(h+1)+1)=m-h$.

{\bf{(2)}}\,\,Let $(h,r)\in\mathbb{N}^{2}$ such that $C$ is $(h,r)$-evasive. Then, we have $h\leqslant r$, and $C$ is $(h-a,r-a)$-evasive for all $a\in[0,h]$.
\end{proposition}

\setlength{\parindent}{2em}
Now we give an alternative characterization of GRWs and profiles for codes over chain rings. The following proposition will be used in Section 4.2.

\begin{proposition}
Suppose that $R$ is a left chain ring and $M$ is a free left $R$-module. Then, for any $l\in[0,m]$, it holds that
\begin{equation}\psi(l)=\max\left(\rank_R(C\cap V^{\Omega})\mid \text{$V$ is a free left $R$-submodule of $M$, $\rank_R(V)=l$}\right).\end{equation}
In addition, there exists a free left $R$-submodule $E\leqslant_{R}M^{\Omega}$ such that $\rank_R(E)=k$ and $C\subseteq E$; moreover, the GRWs and rank profiles of $C$ coincide with those of $E$, respectively.
\end{proposition}

\begin{proof}
Let $l\in[0,m]$, and let $s$ denote the right hand side of (4.8). By definition, we have $s\leqslant\psi(l)$. Conversely, for any $U\leqslant_{R}M$ with $\rank_R(U)=l$, by Lemma 2.4, there exists a free left $R$-submodule $V\leqslant_{R}M$ such that $U\subseteq V$ and $\rank_R(V)=l$; moreover, Lemma 2.4 implies that $\rank_R(C\cap U^{\Omega})\leqslant\rank_R(C\cap V^{\Omega})\leqslant s$, which further implies that $\psi(l)\leqslant s$, as desired. Finally, the rest immediately follows from Proposition 4.1, (4.2) and Lemma 2.4.
\end{proof}

\setlength{\parindent}{2em}
We end this subsection by noting that all the definitions and results in this subsection can be established for finitely generated right modules over a principal right ideal ring in a parallel fashion.

\subsection{Wei-type duality theorem}
\setlength{\parindent}{2em}
Throughout the rest of this section, let $R$ be a chain ring, $m\in\mathbb{Z}^{+}$, $P$ be a finitely generated free left $R$-module with $\rank_R(P)=m$, $Q$ be a finitely generated free right $R$-module with $\rank_R(Q)=m$, and let $\Omega$ be a non-empty finite set with $|\Omega|=n$. Let $f:P\times Q\longrightarrow R$ be a non-degenerated $(R,R)$-bilinear map, and define $\langle~,~\rangle:P^{\Omega}\times Q^{\Omega}\longrightarrow R$ as
\begin{equation}\langle\alpha,\beta\rangle=\sum_{i\in\Omega}f(\alpha_i,\beta_i).\end{equation}
For any $A\subseteq P^{\Omega}$ and $B\subseteq Q^{\Omega}$, let $A^{\bot}\leqslant_{R}Q^{\Omega}$ denote the right dual of $A$ with respect to $\langle~,~\rangle$, and let $^{\bot}B\leqslant_{R}P^{\Omega}$ denote the left dual of $B$ with respect to $\langle~,~\rangle$.

The following Wei-type duality theorem is the main result of this subsection. It generalizes [26, Proposition 65] to codes over chain rings.

\setlength{\parindent}{0em}
\begin{theorem}
Let $C$ be a free left $R$-submodule of $P^{\Omega}$ with $\rank_R(C)=k$. Then, $C^{\bot}$ is a free right $R$-submodule of $Q^{\Omega}$ with $\rank_R(C^{\bot})=mn-k$. Moreover, the following two statements hold:

{\bf{(1)}}\,\,Let $\psi:[0,m]\longrightarrow[0,k]$ and $\eta:[0,m]\longrightarrow[0,mn-k]$ such that $\psi(l)$ and $\eta(l)$ are equal to the $l$-th rank profile of $C$ and $C^{\bot}$, respectively. Then, for any $l\in[0,m]$, it holds that
$$\eta(l)=\psi(m-l)+nl-k;$$

{\bf{(2)}}\,\,Let $\varphi:[0,k]\longrightarrow[0,m]$ such that $\varphi(r)$ is equal to the $r$-th GRW of $C$, and let $\tau:[0,mn-k]\longrightarrow[0,m]$ such that $\tau(s)$ is equal to the $s$-th GRW of $C^{\bot}$. For any $\gamma\in\mathbb{Z}$, define
$$\mathcal{G}_\gamma=\{\varphi(u)\mid u\in[1,k],u\equiv\gamma+k~(\bmod~n)\},$$
$$\mathcal{H}_\gamma=\{m+1-\tau(v)\mid v\in[1,mn-k],v\equiv\gamma~(\bmod~n)\}.$$
Then, for any $\gamma\in\mathbb{Z}$, it holds that $\mathcal{G}_\gamma\cap\mathcal{H}_\gamma=\emptyset$, $\mathcal{G}_\gamma\cup\mathcal{H}_\gamma=[1,m]$.
\end{theorem}

\begin{proof}
{\bf{(1)}}\,\,Let $l\in[0,m]$, and let $T$ denote the set of all the free right $R$-submodules of $Q$ with rank $l$. By (2) of Lemma 2.5, for any $W\in T$, the left dual of $W$ with respect to $f$, denoted by $^{\ddagger}W$, is a free left $R$-submodule $P$ with $\rank_{R}(^{\ddagger}W)=m-l$; moreover, we note that $W^{\Omega}$ is a free right $R$-submodule $Q^{\Omega}$ with $\rank_{R}(W^{\Omega})=nl$, $^{\bot}(W^{\Omega})=(^{\ddagger}W)^{\Omega}$ and $\rank_R(C^{\bot}\cap W^{\Omega})=\rank_R(C\cap (^{\ddagger}W)^{\Omega})+nl-k$. Hence from Proposition 4.6 and Lemma 2.5, we deduce that
\begin{eqnarray*}
\begin{split}
\eta(l)&=\max\left\{\rank_R(C^{\bot}\cap W^{\Omega})\mid W\in T\right\}\\
&=\max\left\{\rank_R(C\cap (^{\ddagger}W)^{\Omega})+nl-k\mid W\in T\right\}\\
&=\max\left\{\rank_R(C\cap V^{\Omega})+nl-k\mid \text{$V$ is a free left $R$-submodule of $P$, $\rank_R(V)=m-l$}\right\}\\
&=\psi(m-l)+nl-k,
\end{split}
\end{eqnarray*}
as desired.

{\bf{(2)}}\,\,This immediately follows from Theorem 4.1, (1) and Theorem 3.12.
\end{proof}

\subsection{MRD codes}
\setlength{\parindent}{2em}
Based on Theorems 4.1 and 4.2, in the rest of this section, we use results established in Section 3 to derive some results on MRD/QMRD/NMRD/$i$-MRD codes and evasive property. We begin with the characterization of MRD codes.

\setlength{\parindent}{0em}
\begin{lemma}
Let $C$ be a free left $R$-submodule of $P^{\Omega}$ with $C\subsetneqq P^{\Omega}$ and $\rank_R(C)=k$, and let $e$ denote the minimal rank weight of $C^{\bot}$. Then, it holds that $m(e-1)\leqslant k$ and $n(e-1)\leqslant k$. Moreover, $C^{\bot}$ is MRD if and only if $k=n(e-1)$.
\end{lemma}

\begin{proof}
Noticing that $\rank_R(C^{\bot})=mn-k$, an application of Corollary 4.1 and (1) of Definition 4.2 to $C^{\bot}$ immediately yields the desired result.
\end{proof}

\setlength{\parindent}{2em}
Now we are ready to prove the main result of this subsection.

\setlength{\parindent}{0em}
\begin{theorem}
Let $C$ be a free left $R$-submodule of $P^{\Omega}$ with $\{0\}\subsetneqq C\subsetneqq P^{\Omega}$ and $\rank_R(C)=k$, and let $d$ and $e$ denote the minimal rank weights of $C$ and $C^{\bot}$, respectively. Then, we have $d+e\leqslant m+2$. Moreover, the following three statements are equivalent to each other:

{\bf{(1)}}\,\,$C$ is MRD;

{\bf{(2)}}\,\,$d+e=m+2$;

{\bf{(3)}}\,\,$C^{\bot}$ is MRD.
\end{theorem}

\begin{proof}
This follows from Lemma 4.2 and Theorem 3.4.
\end{proof}

\begin{remark}
Theorem 4.3 recovers Part 1 of [8, Proposition 19] if $R$ is a field, and also generalizes [4, Theorem 3.23] which has been established for codes over finite commutative chain rings by using MacWilliams identity.
\end{remark}

\subsection{QMRD and dually QMRD codes}

\setlength{\parindent}{2em}
Throughout this subsection, let $C$ be a free left $R$-submodule of $P^{\Omega}$ with $\{0\}\subsetneqq C\subsetneqq P^{\Omega}$ and $\rank_R(C)=k$, and let $d$ and $e$ denote the minimal rank weights of $C$ and $C^{\bot}$, respectively. We say that $C$ is \textit{dually QMRD} if both $C$ and $C^{\bot}$ are QMRD (cf. [8, Definition 13]).

Now we give necessary and sufficient conditions for a code to be dually QMRD. The following theorem is the main result of this subsection.

\setlength{\parindent}{0em}
\begin{theorem}
{\bf{(1)}}\,\,$C^{\bot}$ is QMRD if and only if $n\nmid k$ and $\left\lceil\frac{k}{n}\right\rceil=e$.

{\bf{(2)}}\,\,$C$ is dually QMRD if and only if $d+e=m+1$.
\end{theorem}

\begin{proof}
{\bf{(1)}}\,\,From $\rank_R(C^{\bot})=mn-k$, we deduce that $C^{\bot}$ is QMRD if and only if $n\nmid(mn-k)$ and $\left\lceil\frac{mn-k}{n}\right\rceil=m-e+1$, if and only if $n\nmid k$ and $e=\lfloor\frac{k}{n}\rfloor+1=\lceil\frac{k}{n}\rceil$, as desired.

{\bf{(2)}}\,\,By (1), $C$ is dually QMRD if and only if $n\nmid k$ and $\left\lceil\frac{k}{n}\right\rceil=m-d+1=e$. Hence the desired result immediately follows from Theorem 3.8.
\end{proof}

\setlength{\parindent}{2em}
Next, we characterize MRD or dually QMRD codes in terms of GRWs of $C$. The following proposition is an immediately corollary of Theorems 4.3, 4.4 and Proposition 3.7.

\setlength{\parindent}{0em}
\begin{proposition}
The following three statements are equivalent to each other:

{\bf{(1)}}\,\,$C$ is either MRD or dually QMRD;

{\bf{(2)}}\,\,$d+e\geqslant m+1$;

{\bf{(3)}}\,\,$k\geqslant n(m-d)+1$. Moreover, if $d\leqslant m-1$, then the $(k-n(m-d)+1)$-th GRW of $C$ is equal to $d+1$.
\end{proposition}

\setlength{\parindent}{2em}
We end this subsection by deriving the GRWs of dually QMRD codes. The following proposition follows from Proposition 4.7 and Theorem 3.10.

\begin{proposition}
If $C$ is either MRD or dually QMRD, then for any $a\in[1,k]$, the $a$-th GRW of $C$ is equal to $m-\left\lfloor\frac{k-a}{n}\right\rfloor$, and for any $c\in[1,mn-k]$, the $c$-th GRW of $C^{\bot}$ is equal to $\left\lceil\frac{c+k}{n}\right\rceil$.
\end{proposition}

\begin{remark}
If $R$ is a field, then Theorem 4.4 recovers Part 2 of [8, Proposition 19], and Propositions 4.7 and 4.8 recover [8, Theorem 22 and Remark 23].
\end{remark}

\subsection{NMRD codes}
\setlength{\parindent}{2em}
Now we characterize NMRD codes and compute their GRWs. The following theorem is an immediate corollary of Theorems 3.7, 3.8 and Corollary 3.3.

\setlength{\parindent}{0em}
\begin{theorem}
Let $C$ be a free left $R$-submodule of $P^{\Omega}$ with $\{0\}\subsetneqq C\subsetneqq P^{\Omega}$ and $\rank_R(C)=k$, $d$ and $e$ denote the minimal rank weights of $C$ and $C^{\bot}$, respectively, and let $\theta=\max\{a\in[1,k]\mid\text{the $a$-th GRW of $C$ is equal to $d$}\}$. Then, the following three statements hold true:

{\bf{(1)}}\,\,$C$ is NMRD if and only if $n\mid k$ and $d+e=m$;

{\bf{(2)}}\,\,Suppose that $d+e=m$. Then, for any $a\in[1,k]$, the $a$-th GRW of $C$ is equal to
\begin{eqnarray*}
\begin{split}
\begin{cases}
d,&a\leqslant\theta;\\
m-\left\lfloor\frac{k-a}{n}\right\rfloor,&a\geqslant\theta+1;
\end{cases}
\end{split}
\end{eqnarray*}

{\bf{(3)}}\,\,Suppose that $C$ is NMRD. Then, for any $a\in[\theta+1,k]$, the $a$-th GRW of $C$ is equal to $d+\left\lceil\frac{a}{n}\right\rceil$, and for any $c\in[1,mn-k]$, the $c$-th GRW of $C^{\bot}$ is equal to
\begin{eqnarray*}
\begin{split}
\begin{cases}
e,&c\leqslant\theta;\\
e+\lceil\frac{c}{n}\rceil,&c\geqslant\theta+1.
\end{cases}
\end{split}
\end{eqnarray*}
\end{theorem}

\begin{remark}
If $R$ is a field, then (1) of Theorem 4.5 recovers Part 3 of [7, Theorem 6.9].
\end{remark}

\subsection{$i$-MRD codes}
\setlength{\parindent}{2em}
Now we characterize $i$-MRD codes. The following theorem extends [7, Lemma 6.4 and Theorem 6.5] to codes over chain rings.

\setlength{\parindent}{0em}
\begin{theorem}
Let $C$ be a free left $R$-submodule of $P^{\Omega}$ with $\rank_R(C)=k$, where $1\leqslant k\leqslant n(m-1)$, and let $p$ denote the $(\lceil\frac{k}{n}\rceil n-k+1)$-th GRW of $C^{\bot}$. Then, the following three statements hold:

{\bf{(1)}}\,\,Let $i\in[1,\lceil\frac{k}{n}\rceil]$. Then, $C$ is $i$-MRD if and only if $p\geqslant\lceil\frac{k}{n}\rceil-i+2$;

{\bf{(2)}}\,\,$1\leqslant p\leqslant\lceil\frac{k}{n}\rceil+1$. Moreover, $C$ is $\lceil\frac{k}{n}\rceil$-MRD if and only if $p\geqslant2$;

{\bf{(3)}}\,\,If $p\geqslant2$, then $\lceil\frac{k}{n}\rceil-p+2=\min\{i\in[1,\lceil\frac{k}{n}\rceil]\mid \text{$C$ is $i$-MRD}\}$.
\end{theorem}

\begin{proof}
First, (1) follows from (2) of Theorem 3.6. Next, from $\lceil\frac{k}{n}\rceil n-k+1\geqslant1$ and Theorem 3.3, we deduce that $1\leqslant p\leqslant\lceil\frac{k}{n}\rceil+1$, as desired. Moreover, the rest immediately follows from (1).
\end{proof}

\subsection{Evasiveness of codes}
\setlength{\parindent}{2em}
In this subsection, we study $(h,r)$-evasiveness of codes. We begin with the following theorem, which immediately follows from Theorem 3.11 and Corollary 3.4.

\setlength{\parindent}{0em}
\begin{theorem}
{\bf{(1)}}\,\,Let $C$ be a free left $R$-submodule of $P^{\Omega}$ with $\rank_R(C)=k$, and let $(h,r)\in\mathbb{N}^{2}$ with $h\leqslant r\leqslant m-1$ and $n(h+1)\leqslant k\leqslant n(m+h-r)-1$. Then, $C$ is $(h,r)$-evasive if and only if the $(n(r-h)+1)$-th GRW of $C^{\bot}$ is larger than or equal to $r+2$.

{\bf{(2)}}\,\,Let $C$ be a free left $R$-submodule of $P^{\Omega}$ with $\rank_R(C)=k$, and let $h\in[0,m-2]$ with $n(h+1)\leqslant k\leqslant nm-1$. Then, $C$ is $(h,h)$-evasive if and only if the minimal rank weight of $C^{\bot}$ is larger than or equal to $h+2$.
\end{theorem}

\setlength{\parindent}{2em}
From now on, we focus on $(h,h)$-evasive codes. The following theorem is the main result of this subsection.

\begin{theorem}
Let $C\subsetneqq P^{\Omega}$ be a free left $R$-submodule of $P^{\Omega}$ with $\rank_R(C)=k$, and let $h\in[0,m-2]$ such that $C$ is $(h,h)$-evasive. Then, we have
\begin{equation}m(h+1)\leqslant k.\end{equation}
Moreover, if $m(h+1)=k$, then the minimal rank weight of $C^{\bot}$ is equal to $h+2$.
\end{theorem}

\begin{proof}
Let $e$ denote the minimal rank weight of $C^{\bot}$. By Lemma 4.2 and Theorem 4.7, we have $m(e-1)\leqslant k$ and $e\geqslant h+2$, which imply that $m(h+1)\leqslant m(e-1)\leqslant k$, establishing (4.10). Moreover, if $m(h+1)=k$, then we have $m(h+1)=m(e-1)$ and hence $e=h+2$, as desired.
\end{proof}

\setlength{\parindent}{2em}
The following corollary of Theorem 4.8 gives lower bounds on the rank of $(h,h)$-evasive codes.

\begin{corollary}
Let $D$ be a left $R$-submodule of $P^{\Omega}$ with $\rank_R(D)=k$, and let $h\in[0,m-2]$ such that $D$ is $(h,h)$-evasive. Then, it holds true that either $m(h+1)\leqslant k$ or $k=mn$.
\end{corollary}

\begin{proof}
By Proposition 4.6, there exists a free left $R$-submodule $C\leqslant_{R}P^{\Omega}$ such that $\rank_R(C)=k$ and $D\subseteq C$; moreover, the GRWs of $D$ coincide with those of $C$. It follows that $C$ is $(h,h)$-evasive. Therefore if $C\subsetneqq P^{\Omega}$, then it follows from Theorem 4.8 that $m(h+1)\leqslant k$, and if $C=P^{\Omega}$, then we have $k=mn$, as desired.
\end{proof}

\setlength{\parindent}{2em}
Corollary 4.3 may be regarded as a scattered bound for $(h,h)$-evasive codes as it generalizes the scattered bound for $(h,h)$-evasive subspaces (see [9, Theorem 2.3] and [25, Corollary 4.4]), as detailed in the following remark.

\setlength{\parindent}{0em}
\begin{remark}
Following the settings in (2) of Remark 4.3, let $A\subseteq\mathbb{E}^{[m]}$ be a Gabidulin code with $\dim_{\mathbb{E}}(A)=p$ and $\mathbf{d}_p^{G}(A)=m$, $L\in\Mat_{m,p}(\mathbb{E})$ satisfy that $A=\{L\gamma^{T}\mid\gamma\in\mathbb{E}^{p}\}$, and let $U$ be the $\mathbb{F}$-subspace of $\mathbb{E}^{p}$ generated by all the rows of $L$. Suppose that $C$ is the Delsarte code associated to $A$ with respect of an ordered basis of $\mathbb{E}/\mathbb{F}$. If $h\in[0,m-2]$ and $U$ is $(h,h)$-evasive in $\mathbb{E}^{p}$, then $C$ is $(h,h)$-evasive (see Remark 4.3), and Corollary 4.3 implies that either $\dim_{\mathbb{F}}(U)\leqslant \frac{np}{h+1}$ or $A=\mathbb{E}^{[m]}$ holds true, which recovers [9, Theorem 2.3] and [25, Corollary 4.4].
\end{remark}

\section{Extended generalized poset weight for codes over quasi-Frobenius rings}
\setlength{\parindent}{2em}
The notion of extended generalized Hamming weight (EGHW) was first proposed and studied in the recent paper \cite{32} for linear codes over the ring $\mathbb{Z}_{p^{a}}$, where $p$ is a prime number and $a\in\mathbb{Z}^{+}$. In \cite{32}, the author has established two versions of Wei-type duality theorems which relate the EGHWs of a code with those of its dual code (see [32, Theorems 15 and 16]), extending the classical Wei-type duality theorem for GHWs; moreover, other important properties such as monotonicity and various bounds arising from Singleton bound have also been established (see \cite{32} for more details).

More recently in \cite{33}, the authors introduced the notion of extended generalized poset weight (EGPW) for linear codes over Galois rings, which further generalizes EGHW for Hamming metric codes over $\mathbb{Z}_{p^{a}}$. In \cite{33}, among others, the authors have established monotonicity of EGPWs and two versions of Wei-type duality theorems for EGPWs, and also explored the application of EGPWs in wire-tap channel of type II (see \cite{33} for more details).

In this section, we introduce and study EGPWs and profiles defined for codes over modules which has a composition series. In particular, for codes over quasi-Frobenius rings, we will establish a general Wei-type duality theorem which generalizes and unifies the two versions of Wei-type duality theorems established in both \cite{32} and \cite{33}. Our approach is different from those of \cite{32,33} in the sense that our results are established by using Galois connections, and counterpart results in \cite{32,33} were established by using the explicit structure of modules over $\mathbb{Z}_{p^{a}}$ or Galois rings.

Throughout the rest of this section, let $R$ and $S$ be rings, $M$ be an $(R,S)$-bimodule such that both as left $R$-module and right $S$-module, $M$ has a composition series. Moreover, let $\Omega$ be a non-empty finite set with $|\Omega|=n$. For any $A\subseteq R$ and $B\subseteq S$, define
\begin{equation}A^{\ddag}=\{u\in M\mid \forall~a\in A:au=0\},\end{equation}
\begin{equation}^{\ddag}B=\{u\in M \mid\forall~b\in B:ub=0\}.\end{equation}
For $s\in\mathbb{N}$, a tuple of subsets $(I_0,\dots,I_s)$ of $\Omega$, a tuple of left ideals $(A_0,\dots,A_s)$ of $R$, a tuple of right ideals $(B_0,\dots,B_s)$ of $S$ such that $I_i\subseteq I_{i-1}$, $A_i\subseteq A_{i-1}$, $B_i\subseteq B_{i-1}$ for all $i\in[1,s]$, define
\begin{equation}\rho_{(A_0,\dots,A_s)}(I_0,\dots,I_s)=\{\alpha\in M^{\Omega}\mid \forall~i\in[0,s]:\chi(A_i\alpha)\subseteq I_i\},\end{equation}
\begin{equation}\delta_{(B_0,\dots,B_s)}(I_0,\dots,I_s)=\{\alpha\in M^{\Omega}\mid\forall~i\in[0,s]:\chi(\alpha B_i)\subseteq I_i\}.\end{equation}

The proof of the following preliminary result is straightforward and hence omitted.

\begin{lemma}
Let $s\in\mathbb{N}$, $(I_0,\dots,I_s)$ be a tuple of subsets of $\Omega$, $(A_0,\dots,A_s)$ be a tuple of left ideals of $R$, $(B_0,\dots,B_s)$ be a tuple of right ideals of $S$ such that $I_i\subseteq I_{i-1}$, $A_i\subseteq A_{i-1}$, $B_i\subseteq B_{i-1}$ for all $i\in[1,s]$. Then, we have $\rho_{(A_0,\dots,A_s)}(I_0,\dots,I_s)=\prod_{t\in \Omega}E_t$ and $\delta_{(B_0,\dots,B_s)}(I_0,\dots,I_s)=\prod_{t\in \Omega}F_t$, where for any $t\in\Omega$,
\begin{eqnarray*}
\begin{split}
E_t=\begin{cases}
M,&t\in I_s;\\
{A_0}^{\ddag},&t\in \Omega-I_0;\\
{A_i}^{\ddag},&i\in[1,s],t\in I_{i-1}-I_i,
\end{cases}
\end{split}
\end{eqnarray*}
and
\begin{eqnarray*}
\begin{split}
F_t=\begin{cases}
M,&t\in I_s;\\
^{\ddag}{B_0},&t\in \Omega-I_0;\\
^{\ddag}{B_i},&i\in[1,s],t\in I_{i-1}-I_i.
\end{cases}
\end{split}
\end{eqnarray*}
Consequently, $\len_S(\rho_{(A_0,\dots,A_s)}(I_0,\dots,I_s))$ is equal to
\begin{eqnarray*}
\begin{split}
n\len_S({A_0}^{\ddag})+\left(\sum_{i=0}^{s-1}(\len_S({A_{i+1}}^{\ddag})-\len_S({A_i}^{\ddag}))|I_i|\right)+(\len_S(M)-\len_S({A_s}^{\ddag}))|I_s|,
\end{split}
\end{eqnarray*}
and $\len_R(\delta_{(B_0,\dots,B_s)}(I_0,\dots,I_s))$ is equal to
\begin{eqnarray*}
\begin{split}
n\len_R({^{\ddag}B_0})+\left(\sum_{i=0}^{s-1}(\len_R({^{\ddag}B_{i+1}})-\len_R({^{\ddag}B_i}))|I_i|\right)+(\len_R(M)-\len_R({^{\ddag}B_s}))|I_s|.
\end{split}
\end{eqnarray*}
\end{lemma}

\subsection{Extended generalized poset weights and profiles for left $R$-submodules}

\setlength{\parindent}{2em}
Throughout this subsection, let $C\leqslant_{R}M^{\Omega}$ with $\len_R(C)=k$, $s\in\mathbb{N}$, $(B_0,\dots,B_s)$ be a tuple of right ideals of $S$ with $B_i\subseteq B_{i-1}$ for all $i\in[1,s]$, $\mathbf{P}=(\Omega,\preccurlyeq)$ be a poset, and let
\begin{equation}T=\{(I_0,\dots,I_s)\in\mathcal{I}(\mathbf{P})^{[0,s]}\mid\text{$I_i\subseteq I_{i-1}$ for all $i\in[1,s]$}\}.\end{equation}
We define $\varphi:[0,k]\longrightarrow[0,(s+1)n]$ as
\begin{equation}\varphi(r)=\min\left\{\sum_{i=0}^{s}|\langle\chi(DB_i)\rangle_{\mathbf{P}}|\mid\text{$D\leqslant_{R}C$, $\len_R(D)=r$}\right\},\end{equation}
and define $\psi:[0,(s+1)n]\longrightarrow[0,k]$ as
\begin{equation}\psi(l)=\max\left\{\len_R(C\cap\delta_{(B_0,\dots,B_s)}(I_0,\dots,I_s))\mid(I_0,\dots,I_s)\in T,\sum_{i=0}^{s}|I_i|=l\right\}.\end{equation}

\begin{definition}
For any $r\in[0,k]$, $\varphi(r)$ is referred to as the \textit{$r$-th extended generalized $\mathbf{P}$-weight of $C$ with respect to $(B_0,\dots,B_s)$}, and for any $l\in[0,(s+1)n]$, $\psi(l)$ is referred to as the \textit{$l$-th extended $\mathbf{P}$-profile of $C$ with respect to $(B_0,\dots,B_s)$}.
\end{definition}

\setlength{\parindent}{0em}
\begin{remark}
Definition 5.1 is largely inspired by [33, Definition 3] and [32, Section III, Paragraph 2]. If $s=0$ and $B_0=R$, then (5.6) recovers the generalized $\mathbf{P}$-weights defined for codes over fields and Galois rings (see [27, Section 1] and [33, Definition 2]). Moreover, suppose that $R$ is a chain ring, $\theta\in R$ and $\theta R$ is the unique maximal right ideal of $R$. Since a Galois ring is a chain ring, (5.6) recovers the extended generalized $\mathbf{P}$-weights defined for codes over Galois rings (see [33, Definition 3]) by setting $B_i=\theta^{i}R$ for all $i\in[0,s]$.
\end{remark}

\setlength{\parindent}{2em}
Now we give alternative characterizations of EGPWs and profiles, and show that they form a Galois connection.

\setlength{\parindent}{0em}
\begin{proposition}
For any $r\in[0,k]$, it holds that
\begin{equation}\varphi(r)=\min\left\{\sum_{i=0}^{s}|I_i|\mid(I_0,\dots,I_s)\in T,\len_R(C\cap\delta_{(B_0,\dots,B_s)}(I_0,\dots,I_s))\geqslant r\right\}.\end{equation}
\end{proposition}

\begin{proof}
Let $r\in[0,k]$, and let $q$ denote the right hand side of (5.8). For any $D\leqslant_{R}C$ with $\len_R(D)=r$, we infer that $(L_0,\dots,L_s)\triangleq(\langle\chi(DB_i)\rangle_{\mathbf{P}}\mid i\in[0,s])\in T$ and $D\subseteq C\cap\delta_{(B_0,\dots,B_s)}(L_0,\dots,L_s)$, which implies that $\len_R(C\cap\delta_{(B_0,\dots,B_s)}(L_0,\dots,L_s))\geqslant r$. This implies that $q\leqslant\varphi(r)$. Conversely, for any $(I_0,\dots,I_s)\in T$ with $\len_R(C\cap\delta_{(B_0,\dots,B_s)}(I_0,\dots,I_s))\geqslant r$, we can choose $D\leqslant_{R}C\cap\delta_{(B_0,\dots,B_s)}(I_0,\dots,I_s)$ with $\len_R(D)=r$; moreover, from $D\subseteq\delta_{(B_0,\dots,B_s)}(I_0,\dots,I_s)$, we deduce that $\langle\chi(DB_i)\rangle_{\mathbf{P}}\subseteq I_i$ for all $i\in[0,s]$, which implies that $\sum_{i=0}^{s}|\langle\chi(DB_i)\rangle_{\mathbf{P}}|\leqslant\sum_{i=0}^{s}|I_i|$. It follows that $\varphi(r)\leqslant q$, which further establishes (5.8), as desired.
\end{proof}

\setlength{\parindent}{0em}
\begin{proposition}
For any $l\in[0,(s+1)n]$, it holds that
\begin{equation}\psi(l)=\max\left\{\len_R(C\cap\delta_{(B_0,\dots,B_s)}(I_0,\dots,I_s))\mid(I_0,\dots,I_s)\in T,\sum_{i=0}^{s}|I_i|\leqslant l\right\}.\end{equation}
\end{proposition}

\begin{proof}
Let $l\in[0,(s+1)n]$, and let $q$ denote the right hand side of (5.9). Apparently, we have $\psi(l)\leqslant q$. Conversely, for any $(L_0,\dots,L_s)\in T$ with $\sum_{i=0}^{s}|L_i|\leqslant l$, by [33, Lemma 2], we can choose $(I_0,\dots,I_s)\in T$ such that $\sum_{i=0}^{s}|I_i|=l$ and $L_i\subseteq I_i$ for all $i\in[0,s]$; moreover, it follows from $\delta_{(B_0,\dots,B_s)}(L_0,\dots,L_s)\subseteq\delta_{(B_0,\dots,B_s)}(I_0,\dots,I_s)$ that $\len_R(C\cap\delta_{(B_0,\dots,B_s)}(L_0,\dots,L_s))\leqslant\len_R(C\cap\delta_{(B_0,\dots,B_s)}(I_0,\dots,I_s))\leqslant\psi(l)$, which further implies that $q\leqslant\psi(l)$, as desired.
\end{proof}

\setlength{\parindent}{2em}
Now we are ready to prove the main result of this subsection.

\setlength{\parindent}{0em}
\begin{theorem}
$(\varphi,\psi)$ is a Galois connection between $[0,k]$ and $[0,(s+1)n]$. Moreover, for
$$\mbox{$v\triangleq\max(\{\len_R({^{\ddag}B_{i+1}})-\len_R({^{\ddag}B_i})\mid i\in[0,s-1]\}\cup\{\len_R(M)-\len_R({^{\ddag}B_s})\})$},$$
we have $\psi(l)-\psi(l-1)\leqslant v$ for all $l\in[1,(s+1)n]$, and $\varphi(r)+1\leqslant\max\{\varphi(r+v),1\}$ for all $r\in[0,k-v]$.
\end{theorem}

\begin{proof}
The first assertion follows from Lemma 2.1 and Propositions 5.1 and 5.2. Next, consider $l\in[1,(s+1)n]$. Then, we can choose $(I_0,\dots,I_s)\in T$ such that $\sum_{i=0}^{s}|I_i|=l$ and $\len_R(C\cap\delta_{(B_0,\dots,B_s)}(I_0,\dots,I_s))=\psi(l)$; moreover, by [33, Lemma 2.2], we can choose $(L_0,\dots,L_s)\in T$ such that $\sum_{i=0}^{s}|L_i|=l-1$ and $L_i\subseteq I_i$ for all $i\in[0,s]$. Noticing that there uniquely exists $p\in[0,s]$ such that $|L_p|=|I_p|-1$ and $L_i=I_i$ for all $i\in[0,s]-\{p\}$, from Lemma 5.1, we deduce that
\begin{eqnarray*}
\begin{split}
\psi(l)-\psi(l-1)&\leqslant\len_R(C\cap\delta_{(B_0,\dots,B_s)}(I_0,\dots,I_s))-\len_R(C\cap\delta_{(B_0,\dots,B_s)}(L_0,\dots,L_s))\\
&\leqslant\len_R(\delta_{(B_0,\dots,B_s)}(I_0,\dots,I_s))-\len_R(\delta_{(B_0,\dots,B_s)}(L_0,\dots,L_s))\\
&\leqslant v,
\end{split}
\end{eqnarray*}
as desired. Finally, the rest follows from Lemma 2.2.
\end{proof}

\setlength{\parindent}{0em}
\begin{remark}
Theorem 5.1 recovers monotonicity established for EGHWs and EGPWs (see [32, Theorem 4] and [33, Theorem 2]) when $M=R=S$ and $\len_{R}(B_i)=\len(R)-i$ for all $i\in[0,s]$, where $R$ is set to be either $\mathbb{Z}_{p^{a}}$ or a Galois ring, respectively.
\end{remark}

\subsection{Extended generalized poset weights and profiles for right $S$-submodules}

\setlength{\parindent}{2em}
In this short subsection, we define EGPWs and profiles for right $S$-submodules of $M^{\Omega}$ in a parallel fashion. More precisely, let $E\leqslant_{S}M^{\Omega}$ with $\len_S(E)=k$, $s\in\mathbb{N}$, $(A_0,\dots,A_s)$ be a tuple of left ideals of $R$ with $A_i\subseteq A_{i-1}$ for all $i\in[1,s]$, $\mathbf{P}=(\Omega,\preccurlyeq)$ be a poset, and define $T$ as in (5.5). Moreover, define $\varepsilon:[0,k]\longrightarrow[0,(s+1)n]$ and $\varpi:[0,(s+1)n]\longrightarrow[0,k]$ as
$$\varepsilon(r)=\min\left\{\sum_{i=0}^{s}|\langle\chi(A_iD)\rangle_{\mathbf{P}}|\mid\text{$D\leqslant_{S}E$, $\len_S(D)=r$}\right\},$$
$$\varpi(l)=\max\left\{\len_S(E\cap\rho_{(A_0,\dots,A_s)}(I_0,\dots,I_s))\mid(I_0,\dots,I_s)\in T,\sum_{i=0}^{s}|I_i|=l\right\}.$$
For any $r\in[0,k]$, $\varepsilon(r)$ is referred to as the \textit{$r$-th extended generalized $\mathbf{P}$-weight of $E$ with respect to $(A_0,\dots,A_s)$}, and for any $l\in[0,(s+1)n]$, $\varpi(l)$ is referred to as the \textit{$l$-th extended $\mathbf{P}$-profile of $E$ with respect to $(A_0,\dots,A_s)$}. We note that all the results established in Section 5.1 hold true for $(\varepsilon,\varpi)$ in a Parallel fashion.

\subsection{Wei-type duality theorem}

\setlength{\parindent}{2em}
In this subsection, we establish a Wei-type duality theorem that relates EGPWs and profiles of a code with those of its dual code. From now on, suppose that $R=S$ is a quasi-Frobenius ring, $m\in\mathbb{Z}^{+}$, and $M=R^{[m]}$, the set of all the column vectors over $R$ of length $m$. It is known that the length of a composition series of $R$ as a left $R$-module is equal to the length of a composition series of $R$ as a right $R$-module (see [10, Theorem 58.8]), which we simply denote by $\len(R)$. For any $A\subseteq R$, as in Section 2.2, let $A^{\dag}$ and $^{\dag}A$ denote the right annihilator and left annihilator of $A$ in $R$, respectively. Define the inner product $\langle~,~\rangle:M^{\Omega}\times M^{\Omega}\longrightarrow R$ as
\begin{equation}\langle\alpha,\beta\rangle=\sum_{i\in\Omega}\sum_{j=1}^{m}\alpha_i(j)\beta_i(j),\end{equation}
where $\alpha_i(j)$ denotes the $j$-th entry of $\alpha_i\in R^{[m]}$. For any $D\subseteq M^{\Omega}$, let $D^{\bot}$ and $^{\bot}D$ denote the right dual and left dual of $D$ with respect to $\langle~,~\rangle$, respectively.

\begin{lemma}
Let $s\in\mathbb{N}$, $(B_0,\dots,B_s)$ be a tuple of right ideals of $R$ with $B_0=R$ and $B_i\subseteq B_{i-1}$ for all $i\in[1,s]$, and let $(I_0,\dots,I_s)$ be a tuple of subsets of $\Omega$  with $I_i\subseteq I_{i-1}$ for all $i\in[1,s]$. Moreover, let $(A_0,\dots,A_s)$ be the tuple of left ideals of $R$ such that $A_0=R$ and $A_i={^{\dag}B_{s+1-i}}$ for all $i\in[1,s]$, and define $(J_0,\dots,J_s)$ as $J_i=\Omega-I_{s-i}$ for all $i\in[0,s]$. Then, we have $J_i\subseteq J_{i-1}$, $A_i\subseteq A_{i-1}$ for all $i\in[1,s]$, and it holds that $\delta_{(B_0,\dots,B_s)}(I_0,\dots,I_s)={^{\bot}\rho_{(A_0,\dots,A_s)}(J_0,\dots,J_s)}$.
\end{lemma}

\begin{proof}
It suffices to prove the last assertion. Since $R$ is quasi-Frobenius, we have ${A_i}^{\dag}=B_{s+1-i}$ for all $i\in[1,s]$. Hence from Lemma 5.1, we deduce that ${^{\bot}\rho_{(A_0,\dots,A_s)}(J_0,\dots,J_s)}=\prod_{t\in \Omega}G_t=\delta_{(B_0,\dots,B_s)}(I_0,\dots,I_s)$, where $G_t=\{0\}$ for all $t\in \Omega-I_0$, $G_t=M$ for all $t\in I_s$, and $G_t=(^{\dag}B_{i})^{[m]}$ for all $i\in[1,s]$ and $t\in I_{i-1}-I_{i}$, as desired.
\end{proof}

\setlength{\parindent}{2em}
Now we proceed to state and prove the Wei-type duality theorem. First, we fix some notations. Let $s\in\mathbb{N}$ and $w\in\mathbb{Z}^{+}$ satisfy that
\begin{equation}\len(R)=w(s+1),\end{equation}
$(B_0,\dots,B_s)$ be a tuple of right ideals of $R$ such that $B_i\subseteq B_{i-1}$ for all $i\in[1,s]$ and
\begin{equation}\forall~i\in[0,s]:\len_R(B_i)=(s+1-i)w,\end{equation}
and let $(A_0,\dots,A_s)$ be the tuple of left ideals of $R$ such that $A_0=R$ and
\begin{equation}\forall~i\in[1,s]:A_i={^{\dag}B_{s+1-i}}.\end{equation}
Let $\mathbf{P}=(\Omega,\preccurlyeq_{\mathbf{P}})$ be a poset, $\mathbf{\overline{P}}=(\Omega,\preccurlyeq_{\mathbf{\overline{P}}})$ be the dual poset of $\mathbf{P}$, and let
$$T=\{(I_0,\dots,I_s)\in\mathcal{I}(\mathbf{P})^{[0,s]}\mid\text{$I_i\subseteq I_{i-1}$ for all $i\in[1,s]$}\},$$
$$Y=\{(J_0,\dots,J_s)\in\mathcal{I}(\mathbf{\overline{P}})^{[0,s]}\mid\text{$J_i\subseteq J_{i-1}$ for all $i\in[1,s]$}\}.$$
Let $C$ be a left $R$-submodule of $M^{\Omega}$ with $\len_R(C)=k$, and let $\varphi:[0,k]\longrightarrow[0,(s+1)n]$, $\psi:[0,(s+1)n]\longrightarrow[0,k]$ such that $\varphi(r)$ is equal to the $r$-th extended generalized $\mathbf{P}$-weight of $C$ with respect to $(B_0,\dots,B_s)$, and $\psi(l)$ is equal to the $l$-th extended $\mathbf{P}$-profile of $C$ with respect to $(B_0,\dots,B_s)$. Noticing that $C^{\bot}$ is a right $R$-submodule of $M^{\Omega}$ with $\len_R(C^{\bot})=mn\len(R)-k$, we let $\tau:[0,mn\len(R)-k]\longrightarrow[0,(s+1)n]$ and $\eta:[0,(s+1)n]\longrightarrow[0,mn\len(R)-k]$ such that $\tau(c)$ is equal to the $c$-th extended generalized $\mathbf{\overline{P}}$-weight of $C^{\bot}$ with respect to $(A_0,\dots,A_s)$, and $\eta(l)$ is equal to the $l$-th extended $\mathbf{\overline{P}}$-profile of $C^{\bot}$ with respect to $(A_0,\dots,A_s)$.

\setlength{\parindent}{0em}
\begin{theorem}
{\bf{(1)}}\,\,$(\varphi,\psi)$ is a Galois connection between $[0,k]$ and $[0,(s+1)n]$ satisfying that $\psi(0)=0$ and $\psi(l)-\psi(l-1)\leqslant mw$ for all $l\in[1,(s+1)n]$. Moreover, $(\tau,\eta)$ is a Galois connection between $[0,mn\len(R)-k]$ and $[0,(s+1)n]$.

{\bf{(2)}}\,\,For any $l\in[0,(s+1)n]$, it holds that
$$\eta(l)=\psi((s+1)n-l)+mwl-k.$$
{\bf{(3)}}\,\,For any $\gamma\in\mathbb{Z}$, define
$$\mathcal{G}_\gamma=\{\varphi(u)\mid u\in[1,k],u\equiv\gamma+k~(\bmod~mw)\},$$
$$\mathcal{H}_\gamma=\{(s+1)n+1-\tau(v)\mid v\in[1,mn\len(R)-k],v\equiv\gamma~(\bmod~mw)\}.$$
Then, for any $\gamma\in\mathbb{Z}$, it holds that $\mathcal{G}_\gamma\cap\mathcal{H}_\gamma=\emptyset$, $\mathcal{G}_\gamma\cup\mathcal{H}_\gamma=[1,(s+1)n]$.
\end{theorem}

\begin{proof}
{\bf{(1)}}\,\,From $B_0=R$ and Lemma 5.1, we deduce that $\delta_{(B_0,\dots,B_s)}(\emptyset,\dots,\emptyset)=\{0\}$, which implies that $\psi(0)=0$, as desired. Next, for any $i\in[0,s]$, by Lemma 2.5, we have $\len_{R}({^{\dag}B_i})=\len(R)-\len_{R}(B_i)=wi$ and hence $\len_R({^{\ddag}B_i})=\len_R(({^{\dag}B_i})^{[m]})=mwi$. Moreover, by (5.11), we have $\len_R(M)=mw(s+1)$. Hence the desired result follows from Theorem 5.1.

{\bf{(2)}}\,\,Let $l\in[0,(s+1)n]$. Since $\mathcal{I}(\mathbf{\overline{P}})=\{\Omega-K\mid K\in\mathcal{I}(\mathbf{P})\}$ (see [20, Lemma 1.2]), $(I_0,\dots,I_s)\mapsto(\Omega-I_{s-i}\mid i\in[0,s])$ induces a bijection from $\{(I_0,\dots,I_s)\in T\mid\sum_{i=0}^{s}|I_i|=(s+1)n-l\}$ to $\{(J_0,\dots,J_s)\in Y\mid\sum_{i=0}^{s}|J_i|=l\}$. Now let $(I_0,\dots,I_s)\in T$ with $\sum_{i=0}^{s}|I_i|=(s+1)n-l$, and define $(J_0,\dots,J_s)\in Y$ as $J_i=\Omega-I_{s-i}$. For any $i\in[1,s]$, by (5.13), we have ${A_i}^{\dag}=B_{s+1-i}$, which, together with (5.12), implies that $\len_{R}({A_i}^{\ddag})=\len_{R}({B_{s+1-i}}^{[m]})=mwi$. Hence from Lemma 5.1, we deduce that $\len_{R}(\rho_{(A_0,\dots,A_s)}(J_0,\dots,J_s))=mwl$, which, together with Lemma 5.2 and (1) of Lemma 2.5, further implies that
$$\len_{R}(C\cap\delta_{(B_0,\dots,B_s)}(I_0,\dots,I_s))-\len_{R}(C^{\bot}\cap\rho_{(A_0,\dots,A_s)}(J_0,\dots,J_s))=k-mwl.$$
From the above discussion, we deduce that
\begin{eqnarray*}
\begin{split}
\eta(l)&=\max\left\{\len_{R}(C^{\bot}\cap\rho_{(A_0,\dots,A_s)}(J_0,\dots,J_s))\mid(J_0,\dots,J_s)\in Y,\sum_{i=0}^{s}|J_i|=l\right\}\\
&=\max\left\{\len_{R}(C\cap\delta_{(B_0,\dots,B_s)}(I_0,\dots,I_s))+mwl-k\mid(I_0,\dots,I_s)\in T,\sum_{i=0}^{s}|I_i|=(s+1)n-l\right\}\\
&=\psi((s+1)n-l)+mwl-k,
\end{split}
\end{eqnarray*}
as desired.

{\bf{(3)}}\,\,This follows from (1), (2) and Theorem 3.12.
\end{proof}

\setlength{\parindent}{0em}
\begin{remark}
If $s=0$ and $M=R$, then Theorem 5.3 recovers the Wei-type duality theorem for GHWs and GPWs; moreover, if $w=1$ and $M=R$, then Theorem 5.3 recovers the Wei-type duality theorem for EGHWs and EGPWs (see [32, Theorems 15, 16] and [33, Theorems 8, 9]).
\end{remark}

\end{document}